\documentclass[acmsmall,screen]{acmart}
\usepackage{xspace,graphicx,color,cleveref}
\usepackage{float}
\usepackage{mathpartir}
\usepackage{tikz}
\usetikzlibrary{trees,arrows.meta,positioning}
\usepackage{graphicx}
\usepackage{subcaption}

\newcommand{\name}{Chordata\xspace}
\newcommand{\todo}[1]{{\color{red}#1\xspace}}

\setcopyright{acmlicensed}
\copyrightyear{2026}
\acmYear{2026}
\acmDOI{XXXXXXX.XXXXXXX}

\begin{document}

\title{Incremental Live Programming via Shortcut Memoization}

\author{Marisa Kirisame\textsuperscript{*}}
\email{marisa@cs.utah.edu}
\author{Thomas J. Porter\textsuperscript{*}}
\email{thomasjp@umich.edu}
\author{Ruqing Yang}
\email{ryangaq@connect.ust.hk}
\author{Jianqiu Zhao}
\email{qiuscau@stu.scau.edu.cn}
\author{Yudi Wu}
\email{9680384274@g.ecc.u-tokyo.ac.jp}
\author{Ivan Wei}
\email{ivanwei@umich.edu}
\author{Cyrus Omar}
\email{comar@umich.edu}
\author{Pavel Panchekha}
\email{pavpan@cs.utah.edu}

\newcommand{\Variable}{\texttt{Variable}}
\newcommand{\Control}{\texttt{Control}}
\newcommand{\Nil}{\texttt{Nil}}
\newcommand{\Cons}{\texttt{Cons}} 
\newcommand{\Match}{\texttt{Match}} 
\newcommand{\Value}{\texttt{Value}}
\newcommand{\Env}{\texttt{Env}}
\newcommand{\Kont}{\texttt{Kont}}
\newcommand{\OK}{\texttt{OK}}
\newcommand{\State}{\texttt{State}}
\newcommand{\Eval}{\texttt{Eval}}
\newcommand{\Apply}{\texttt{Apply}}

\newcommand{\Constructor}{\texttt{Constructor}}
\newcommand{\Term}{\texttt{Term}}
\newcommand{\PatternVar}{\texttt{PatternVar}}
\newcommand{\Pattern}{\texttt{Pattern}}
\newcommand{\Rule}{\texttt{Rule}}
\newcommand{\Substitution}{\texttt{Substitution}}

\newcommand{\Let}{\texttt{Let}} 
\newcommand{\In}{\texttt{In}} 
\newcommand{\Int}{\texttt{Int}}
\newcommand{\Type}{\texttt{Type}}
\newcommand{\Function}{\texttt{Function}}
\newcommand{\DataType}{\texttt{DataType}}
\newcommand{\Definition}{\texttt{Definition}}
\newcommand{\Add}{\texttt{Add}}

\newcommand{\TotalDegree}{\texttt{total\_degree}}
\newcommand{\MaxDegree}{\texttt{max\_degree}}
\newcommand{\hazelPointCount}{4615}
\newcommand{\hazelTotalSpeedup}{$13.03\times$}
\newcommand{\hazelTotalMemoryOverhead}{$19.97\times$}
\newcommand{\hazelSpeedBreakdown}{%
\begin{tabular}{lr}
\hline
Name & Percent \\
\hline
compose rule & 28.97\% \\
lookup rule & 28.88\% \\
intantiation & 18.36\% \\
insert rule & 12.87\% \\
apply rule & 10.07\% \\
misc & 0.86\% \\
\hline
\end{tabular}%
}
\newcommand{\hazelSpeedupTable}{%
\begin{tabular}{l|rrr|rrr}
\hline
Benchmark & Alice & Bob & Cam & Alice & Bob & Cam \\
 & \multicolumn{3}{c|}{time speedup} & \multicolumn{3}{c}{memory overhead} \\
\hline
Append & 18.2 & 19.43 & 22.43 & 6.931 & 12.16 & 14.88 \\
Filter & 2.558 & 11.49 & 11.32 & 12.73 & 8.693 & 10.29 \\
Map & 12.64 & 4.564 & 7.209 & 5.674 & 9.07 & 7.654 \\
QuickSort & 15.14 & 8.868 & 8.91 & 48.2 & 37.1 & 73.98 \\
InsertSort & 13.27 & 18.56 & 16.99 & 58.58 & 28.85 & 33.39 \\
MergeSort & 5.079 & 15.7 & X & 121.1 & 186.9 & X \\
Pair & 8.231 & 12.7 & X & 9.399 & 11.03 & X \\
Reverse & 4.775 & 12.83 & 21.45 & 29.95 & 17.17 & 8.007 \\
\hline
\end{tabular}%
}

\newcommand{\arithTotalSpeedup}{$2.464\times$}

\begin{abstract}
Live programming systems 
  aim to quickly show programmers the dynamic impacts of program edits.
To do so, they re-execute the program whenever it is edited, 
  which poses a computational challenge
  when programs become large or complex.
  This has led to the need for incrementality 
  in the implementation of live program interpreters. 
This paper introduces \name, an incremental program interpreter
  based on \emph{shortcut memoization},
  which learns repeated patterns of computation,
  called shortcuts, by observing executions of previous versions of a program.
It can then apply these shortcuts
  when the same or a \emph{structurally similar} program fragment is re-executed.
This paper contributes a formal semantics of shortcut memoization
  for any language with a rewrite-based semantics, 
  with mechanized proofs of key correctness properties. 
We then express a variant of the Hazel live programming system, 
  expressed as a CEK machine, in \name,
  and develop a number of practical heuristics 
  to learn high-value shortcuts. 
We evaluate the resulting system on editing traces
  of students solving simple programming problems.
\name achieves a speedup of \hazelTotalSpeedup{} compared to baseline with a \hazelTotalMemoryOverhead{} memory overhead.
For smaller changes and for more complex programs, \name achieves even greater speedups.
Furthermore, we show that \name is capable of providing a speedup 
  even within a single execution, with a faster speedup on a larger input.
\end{abstract}



\maketitle
\footnotetext{\textsuperscript{*}These authors contributed equally to this work.}

\section{Introduction}
Live programming environments, which interleave program editing with execution~\cite{tanimoto}, have a long and storied history. 
  Spreadsheets pioneered live feedback in response to edits~\cite{spreadsheets}, and there has been a long line of work on other systems with varying degrees of liveness, including modern systems such as Hazel (functional programming)~\cite{popl19}, Jupyter (data science)~\cite{ipython} and Vercel (web development)~\cite{vercel}. 

For an environment to feel truly live,
  the user must see the results of their edits quickly;
  long latency windows cause second-guessing and mistakes,
  and disrupt the user's focus~\cite{vlhcc25}.
This is especially true in systems like Hazel, 
  where the program is executed on the fly with every keystroke, even when there are holes and errors in the program (\Cref{fig:hazel}).

\begin{figure}
    \includegraphics[width=0.49\linewidth]{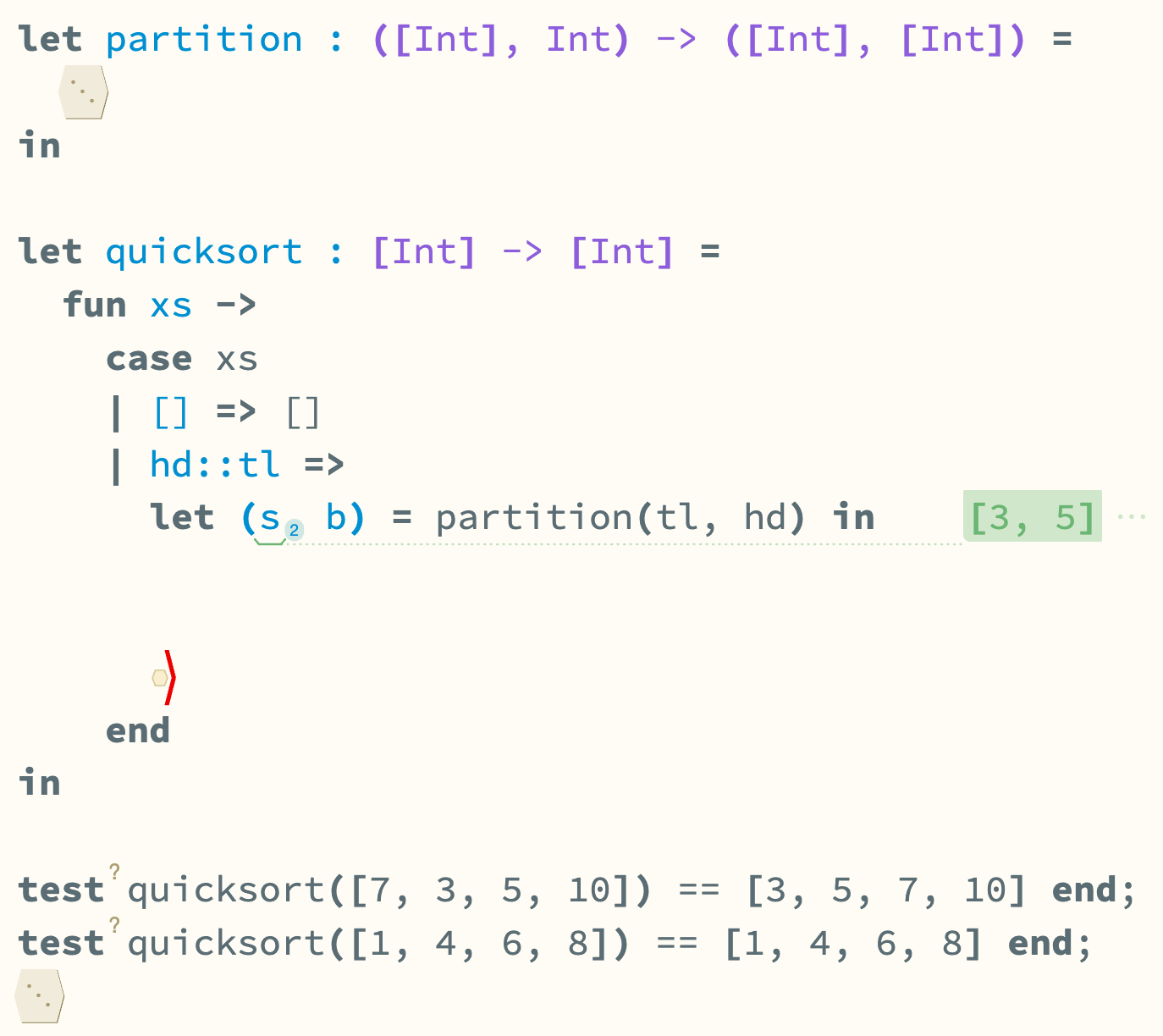}
    \includegraphics[width=0.49\linewidth]{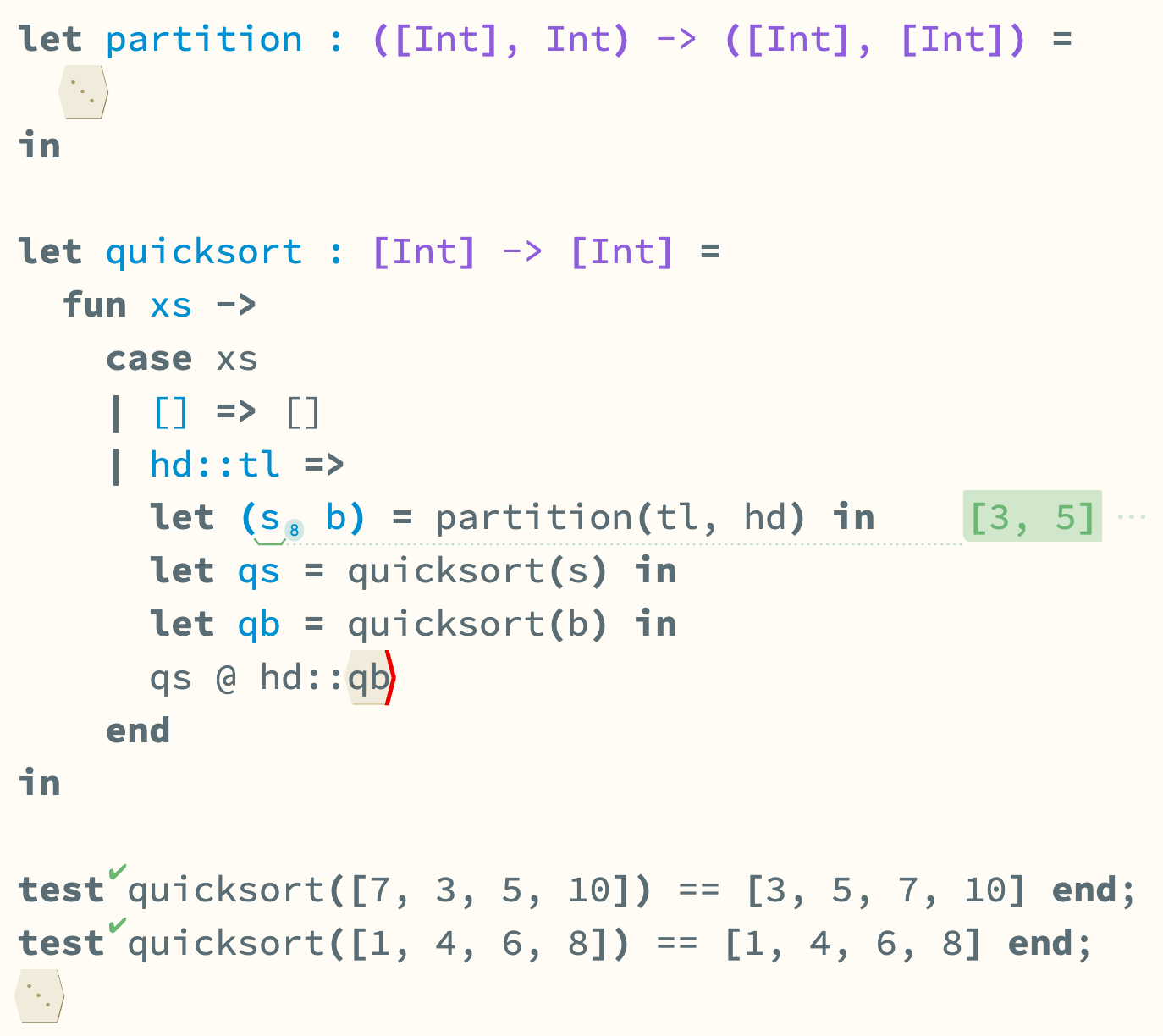}
    \caption{Screenshot of implementing \texttt{quicksort} in the Hazel live programming environment. While the user is editing the function, the program is incomplete (the hole at the red cursor on the left) and/or ill typed (not shown).
    The user can insert a probe to receive live feedback about a component of the result of calling \texttt{partition}, here \texttt{s}. As the user fills in the remaining function (right), Hazel will re-execute the program on every keystroke, updating the probe (here, showing that there are 8 total probed values available rather than the initial 2 from the tests shown). Currently, Hazel does a full re-execution on every keystroke, which can have high latency even for small programs, e.g. when there are many tests being run (elided here). The aim of \name is to speed up these re-executions by learning shortcuts from prior executions.}
    \Description{Two Hazel screenshots while editing quicksort.
    The left screenshot shows an incomplete function with a hole at
    the cursor and a probe displaying intermediate results.
    The right screenshot shows a more complete version of the
    program with more probe results. The figure illustrates that
    Hazel re-executes the program on each edit.}
    \label{fig:hazel}
\end{figure}

An approach to address this latency challenge
  is with \emph{incremental computation}.
Across different areas of computing,
 there is often a need to cache information from earlier executions
  that can be reused to speed up later executions of structurally similar programs (or the same program with similar input).
This general approach is seen, for example,
  in differential dataflow in databases~\cite{differential-dataflow},
  incremental layout in web browsers~\cite{incremental-layout},
  incremental parsing~\cite{incremental-parsing},
  type checking~\cite{incremental-type-checking},
  compilation~\cite{incremental-compilation}, 
  and even linking in programming systems~\cite{incremental-linking}. 
In live programming systems, there is an acute need for a general-purpose
  incremental program interpreter that is able to exploit the structure 
  shared across edits to speed up execution.

Many incremental programming systems, 
  including spreadsheets~\cite{spreadsheets} 
  and build systems~\cite{build-systems-ala-carte}, 
  build a \emph{dependency graph} between coarsely distinguished program structures (e.g. modules, or cells in a notebook or spreadsheet) 
  that then drives incrementality. 
Edits cause recomputation to occur only for structures downstream along the dependency graph. However, there are two main problems with this approach. 

First, in live programming systems, many edits happen locally within a single structure and thus do not benefit from this approach (the \textbf{granularity problem}). For example, in \autoref{fig:hazel}, all edits occur within the \texttt{quicksort} function.

Second, these approaches determine dependencies 
  from the syntactic structure of the computation
  and thus struggle with sequences of edits 
  that leave the program syntactically quite different.
Often, semantic similarity is not exploited
  with spatially and temporally distal code elements
  (the \textbf{distal similarity problem}). 
For example, 
  edits may involve refactoring expressions into auxiliary functions, 
  lifting sub-expressions into bound variables, 
  removing duplicate code, or temporarily removing code. 
Edits may also involve optimizing code into an equivalent form,
  e.g. replacing $filter(\lambda\_.\texttt{true}, x)$ with $x$. 
We would like to be able to exploit these distal semantic similarities, 
  even when the syntax has been modified heavily.


\subsection{Shortcut Memoization}
This paper addresses these problems
  by introducing \emph{shortcut memoization},
  an approach that learns shortcuts from previous runs
  and reuses them when the machine encounters a similar program state again.
Specifically, shortcut memoization applies
  to abstract machines that execute programs
  through a step relation between states $\Sigma~\mapsto~\Sigma'$.
Shortcuts are stepping patterns
  $\Sigma[\bar{x}] \mapsto \Sigma'[\bar{x}]$
  that capture multiple steps of program execution.
The variables in these patterns abstract over uninspected portions of state,
  so they match more generally than just a single trigger state.

In our implementation of shortcut memoization, which we call \name,
  we fix the abstract machine to a CEK machine \cite{cek}.
We then develop heuristics
  to prune the space of possible shortcuts,
  selecting those that match often
  and skip many program states.
We use a discrimination tree approach 
  to efficiently match the corpus of learned shortcuts
  to the current program state.

Because \name learns shortcuts over machine states 
  rather than tracking binding dependencies, 
  reuse is determined by the structure of execution 
  rather than the syntactic form of the program. 
Consequently, shortcut memoization addresses both the 
  \textbf{granularity problem}
  and the \textbf{distal similarity problem}. 
We handle both program code and input data uniformly 
  in the machine state, 
  so \name naturally handles changes to programs and data. 
This makes \name particularly attractive 
  for live programming environments, 
  where programs change rapidly but often only in small ways.
In our evaluation, we show that \name is able to achieve
  speed-ups of up to \hazelTotalSpeedup
  \ on a selection of live program editing traces derived from Hazel users, 
  while paying \hazelTotalMemoryOverhead~ memory overhead.
We demonstrate in \Cref{sec:eval} that
  rewrite rules can also speed up individual executions,
  e.g. when they operate over repetitive data.

\subsection{Motivating example}
To explain the key mechanics of \name,
  consider the following bit-flipping function:
\[
    f([]) = [] \hspace{.5in}
    f(0 :: \ell) = 1 :: f(\ell) \hspace{.5in}
    f(1 :: \ell) = 0 :: f(\ell)
\]
Think about the mechanical, step by step execution
  of this function:
  every application of $f$ (except the last one)
  first tests the identity of the head of the list
  and then recurses by pushing a stack frame.
The stack ends up mirroring the original list,
  but reversed, with the head of the stack
  representing the last element of the list
  and each return address on the stack
  indicating which branch $f$ took,
  which means whether the list had a $0$ or $1$ element.
Then, when the tail of the list is reached,
  the output list is built
  by popping stack frames from the stack
  and constructing a list cell for each one.
In short, during execution, this program has two states---%
  one when it is recursing down the list,
  and another when it is constructing the result---%
  and each state has both a current list value
  and a current stack of stack frames.

\newcommand{\Bool}{\texttt{Bool}}
\newcommand{\True}{\texttt{True}}
\newcommand{\False}{\texttt{False}}

\newcommand{\List}{\texttt{List}}
\newcommand{\Stack}{\texttt{Stack}}
\newcommand{\Recurse}{\texttt{Recurse}}
\newcommand{\Pop}{\texttt{Pop}}
\newcommand{\Done}{\texttt{Done}}
\newcommand{\Empty}{\texttt{Empty}}

We can represent this low-level, mechanical view
  of program execution like so:
\[\begin{array}{lcllll}
\ell & \in & \List & \Coloneqq & [] \mid 0 :: \ell \mid 1 :: \ell\\
k & \in & \Stack & \Coloneqq & \Empty \mid \Cons1~k \mid \Cons0~k \\
s & \in & \State & \Coloneqq & \Recurse~\ell~k \mid \Pop~\ell~k \mid \Done~\ell
\end{array}\]

The initial execution $f(\ell)$
  corresponds to the state $\Recurse~\ell~\Empty$.
In this \Recurse{} state, execution recurses down the list:
\[\begin{array}{lcl}
    \Recurse~(0 :: \ell)~k & \mapsto & \Recurse~\ell~(\Cons1~k) \\
    \Recurse~(1 :: \ell)~k & \mapsto & \Recurse~\ell~(\Cons0~k)
\end{array}\]

Then, when execution reaches the end of the list,
  it switches to the \Pop{} state,
  which pops from the stack and constructs
  the output list $\ell_2$:
\[\begin{array}{lcl}
    \Recurse~\Nil~k & \mapsto & \Pop~\Nil~k \\
    \Pop~\ell~(\Cons0~k)
      & \mapsto & \Pop~(0 :: \ell)~k \\
    \Pop~\ell~(\Cons1~k)
      & \mapsto & \Pop~(1 :: \ell)~k \\
    \Pop~\ell~\Empty & \mapsto & \Done~\ell
\end{array}\]

Finally, the machine reaches the $\Done$ state,
  which represents the end of $f$'s execution.

An important property
  of this operational semantics
  is that it is defined by \emph{pure} rewrite rules.
Each step in the $\Sigma \mapsto \Sigma'$ relation
  has no premises or recursive sub-steps,
  so execution in these semantics consists solely of
  finding a rule whose left-hand side
  matches the current state
  and stepping to the new state
  given by the rule's right-hand side.
In fact, the variables
  on each rule's left-hand side are always distinct,
  and each rule's right-hand side only contains variables
  from the left hand side.
  Because of this, execution conceptually involves only
  matching constructors and binding variables.

Moreover, these rewrite rules can be composed.
For example, any state 
  $\Recurse~(1 :: 0 :: \ell)~k$
  will step first to $\Recurse~(0 :: \ell)~(\Cons0~k)$
  and then to $\Recurse~\ell~(\Cons1~(\Cons0~k))$,
  so the following rewrite rule,
  while not part of the operational semantics above,
  is also valid: 
\[
  \Recurse~(1 :: 0 :: \ell)~k \mapsto \Recurse~\ell~(\Cons1~(\Cons0~k))
\]
This new rewrite rule, which we call a \textit{shortcut},
  summarizes two execution steps in a single rule.
Adding it to the operational semantics
  does not change the meaning of any program,
  while also executing programs in fewer steps. 
\name does exactly this,
  deriving shortcuts by observing program execution
  and then extending its operational semantics
  with these learned rules. 

\subsection{Contributions}
\noindent
In short, \name is an incremental interpreter
  based on shortcuts rather than dependency graphs.
Its core contributions are as follows:
\begin{itemize}
\item A formal semantics for creating shortcuts, with a mechanized proof (in Agda) that rule composition is correct and most general (\Cref{sec:semantics}).
\item Heuristics for learning high-value shortcuts (\Cref{sec:learn})
\item An efficient insert/lookup procedure for matching shortcuts (\Cref{sec:lookup})
\item A data representation to speed up shortcut applications (\Cref{subsec:data-rep})
\item An optimized discrimination tree data structure to manipulate complex rewrite rules {(\Cref{subsec:data-rep})}
\item A collection of additional optimizations incorporated into \name (\Cref{sec:opt})
\item An evaluation of Chordata on live program traces from Hazel that demonstrates substantial speedups (\Cref{sec:eval})
\end{itemize}

\section{Formalization}
\label{sec:semantics}

The \name approach is applicable to
  any language of syntax trees
  endowed with small-step reduction semantics
  defined by a set of pattern-matching rules.
This section first shows an applicable type of semantics---%
  CEK machines---and then generalizes to arbitrary
  syntax trees and reduction semantics.%
\footnote{
Traditional small-step semantics
  involves environments, substitutions, 
  and recursive walks to find reducible expressions,
  which are not local pattern-matches
  and so not a good match for \name.
}
It then defines the composition operation
  that \name uses to find shortcuts,
  and proves that shortcuts exist and are valid:
  that is, that their application yields
  the same result as non-incremental evaluation.
  Every theorem in this section is justified by a mechanized proof in the Agda proof assistant (included in the supplemental material). 

\subsection{CEK Abstract Machine}
\label{sec:abstract-machine}

The CEK abstract machine \cite{cek}
  is a way to define a reduction semantics
  for a functional programming language
  via a set of rewrite rules
  on syntax trees over a set of constructors.
This makes it perfect for \name's shortcut memoization.
In the CEK machine,
  the language semantics is a step relation on
  a machine state consisting of three parts:
  the \textit{C}ontrol expression to evaluate,
  the \textit{E}nvironment binding variables to values,
  and the \textit{K}ontinuation to return to,
  each of which has many possible constructors.
For example, the following defines a CEK machine
  for a minimal language of list construction and destruction:

\[\begin{array}{lcllll}
x & \in & \Variable \\
v & \in & \Value & \Coloneqq & \Nil \mid \Cons~v~v \\
c & \in & \Control & \Coloneqq & x \mid \Nil \mid \Cons~c~c \mid \Match~c~[\Nil \rightarrow c]~[\Cons~x~x \rightarrow c] \\
e & \in & \Env & \Coloneqq & \overline{x:=v} \\
k & \in & \Kont & \Coloneqq & \OK \mid \Cons\_~c~e~k \mid \Cons~v~\_~k \mid \Match\_~[\Nil \rightarrow c]~[\Cons~x~x \rightarrow c]~e~k \\
s & \in & \State & \Coloneqq & \Eval~c~e~k \mid \Apply~v~k
\end{array}\]

Note that the $\Apply$ states,
  while non-standard,
  are folklore in the literature
  and ensure that the step relation is defined
  entirely by a fixed set of rewrite rules.
In effect, these extra states unroll recursive traversals
  into multiple atomic steps.%
\footnote{
The rules for $\Eval~x$ seem to rely on the inequality
  between $x$ and $y$;
  however, in a real program,
  only a finite number of variables are present,
  so these two rules can be expanded to a
  finite number of pure rewrite rules.
}

Although the language in this case is minimal,
  CEK machines can be defined for most functional languages. 
The execution of a CEK machine program is defined by
  a step relation $s \mapsto s'$ over states.
Critically, in a CEK machine,
  the entire step relation is defined by
  a finite number of simple patterns:
\[\begin{array}{lcl}
    \Eval~x~[y := v; e]~k & \mapsto & \Eval~x~e~k \\
    \Eval~x~[x := v; e]~k & \mapsto & \Apply~v~k \\[6pt]

    \Eval~\Nil~e~k & \mapsto & \Apply~V\Nil~k \\
    \Eval~(\Cons~c_1~c_2)~e~k & \mapsto & \Eval~c_1~e~(\Cons~\_~c2~e~k) \\
    \Apply~v_1~(\Cons~\_~c_2~e~k) & \mapsto & \Eval~c_2~e~(\Cons~v_1~\_~k) \\
    \Apply~v_2~(\Cons~v_1~\_~k) & \mapsto & \Apply~(\Cons~v_1~v_2)~k \\[6pt]

    \Eval~(\Match~c_1~\cdots)~e~k & \mapsto & \Eval~c_1~e~(\Match~\_~\cdots~e~k) \\
    \Apply~\Nil~(\Match~\_~[\Nil \rightarrow c1]~\cdots~e~k) & \mapsto & \Eval~c1~e~k \\
    \Apply~(\Cons~v_1~v_2)~(\Match~\_~\cdots~[\Cons~x_1~x_2 \rightarrow c_2]~e~k) & \mapsto & \Eval~c_2~[x_1 := v_1; x_2 := v_2; e]~k
\end{array}\]

There is no rule for applying a value to the $\OK$ continuation
  or for an $\Eval~x$ in an empty environment;
  the former represents successful program termination
  and the latter represents an error state.

\subsection{Rewrite Rules}
More generally, \name's input is
  a programming language definition
  given by a sort \Constructor~of constructors
  and a set of rewrite rules
  that define its reduction semantics.
For the CEK machine,
  the set of constructors must include each constructor contained in the definitions of
  \Control, \Value, \Env, \Kont, and \State,
  while the set of rewrite rules is exactly the set displayed above.
We define a language of terms, 
  each of which is a constructor applied to a (possibly empty) sequence of arguments:
\[\begin{array}{lcllll}
k & \in & \Constructor \\
t & \in & \Term & \Coloneqq & k~\overline{t}
\end{array}\]

Rewrite rules over these terms are defined by patterns,
  which are terms where some subterms
  are replaced with \textit{pattern variables}:

\[\begin{array}{lcllll}
x & \in & \PatternVar \\
p & \in & \Pattern & \Coloneqq & k~\overline{p} \mid x
\end{array}\]

We assume that the set of pattern variables $\PatternVar$ is infinite. In particular, we assume that there is an injection from $\mathbb{N}$ to $\PatternVar$, and a bijection between $\PatternVar$ and $\PatternVar + \PatternVar$ (where $+$ denotes disjoint union). These operations are used to generate fresh variables. 
A reduction rule then consists of two patterns,
  a right hand side and a left hand side, written $p \mapsto p$: 

\[\begin{array}{lcllll}
r & \in & \Rule & \Coloneqq & p \mapsto p
\end{array}\]

We make the additional assumption that every rule is \textit{functional}, that is, that every variable appearing on its right hand side (RHS) 
  also appears on its left hand side (LHS). 
The input to \name is
  a sort of constructors $\Constructor$ and
  a set of rules $R\subseteq \Rule$.
These rules define the semantics of the language via
  a single inference rule
\begin{equation}
\label{eq:therule}
\inferrule{
        r = p_1 \mapsto p_2\\ 
        \sigma[p_1] = t_1\\
        \sigma[p_2] = t_2\\
    }{
        t_1 \overset{r}{\mapsto} t_2
    }{}
\end{equation}

$t_1 \overset{r}{\mapsto} t_2$ is pronounced
  ``$t_1$ steps to $t_2$ by $r$''.
The inference rule states that
  one term steps to another by some reduction rule
  if there exists a substitution $\sigma$
  that can be applied to each side of the reduction rule
  to obtain each term.
Note that for the second and third premises to make sense,
  we choose to consider terms as special patterns:
  those with no pattern variables. 

Terms $t_1$ that do not step to any other term by any rule
  are considered final, or terminating, program states.
If $t_1 \overset{r}{\mapsto} t_2$ and $r\in R$,
  we write $t_1 \overset{R}{\mapsto} t_2$,
  and define $\overset{R}{\mapsto}^*$ as
  the reflexive transitive closure of $\overset{R}{\mapsto}$.
We say that $t_1$ evaluates to $t_2$ under $R$ (written $t_1 \overset{R}{\Rightarrow} t_2$)
  iff $t_1 \overset{R}{\mapsto}^* t_2$,
  and $t_2 \not\overset{R}{\mapsto} t_3$ for every $t_3$.

The assumptions of \Cref{eq:therule}
  refer to substitutions, which we now define.

\subsection{Substitution and Unification}

Substitutions are defined as mappings from pattern variables to patterns. 
It will be convenient for the rest of this section to consider terms 
  to be just patterns containing no pattern variables 
  and to reason generally about patterns rather than 
  specifically about terms. 
Substitutions can be applied to patterns (written $\sigma[p]$) 
  by replacing each pattern variable with its corresponding pattern. 

\[\begin{array}{lcllll}
\sigma & \in & \Substitution & := \PatternVar \rightarrow \Pattern\\[6pt]

&& \sigma[k~\overline{p}]&= k~\sigma[\overline{p}]\\
&& \sigma[x]&= \sigma(x)\\
&& \sigma[p \overline{p}]&= \sigma[p] \sigma[\overline{p}]\\
&& \sigma[\cdot]&= \cdot
\end{array}\]

We now define unifier and most general unifiers for a pair of patterns. 
Note that we consider pairs of substitutions to unify pairs of patterns, rather than the more standard notion of a single substitution 
  unifying some set of constraints. 

\begin{definition}
    $\sigma_1 \circ \sigma_2$ is defined such that $(\sigma_1 \circ \sigma_2)(x) = \sigma_1[\sigma_2(x)]$ for all $x$.
\end{definition}

\begin{definition}
    $\sigma_1 \sqsubseteq \sigma_2$ iff $\sigma_1 = \sigma_3 \circ \sigma_2$ for some $\sigma_3$.
\end{definition}

\begin{definition}
    $(\sigma_1, \sigma_2)$ \textit{unifies} a pair of patterns $(p_1, p_2)$ iff $\sigma_1[p_1] = \sigma_2[p_2]$. 
\end{definition}

\begin{definition}
    $(\sigma_1, \sigma_2)$ is a \textit{most general unifier} for a pair of patterns $(p_1, p_2)$ iff $(\sigma_1, \sigma_2)$ unifies $(p_1, p_2)$ and for all $(\sigma_1', \sigma_2')$ that unifies $(p_1, p_2)$, $\sigma_1' \sqsubseteq  \sigma_1$ and $\sigma_2' \sqsubseteq  \sigma_2$. 
\end{definition}

The notion of the most general unifier is critical to defining the composition of rules.

\subsection{Shortcuts}
\label{sec:shortcut-theory}

Now that substitution and unification are defined,
  we can introduce the concept of shortcuts,
  which allow multiple reduction steps to be combined into one.
Consider a sequence of two reduction steps:

\[
t_1 \overset{r_1}{\mapsto} t_2 \overset{r_2}{\mapsto} t_3
\]

We desire a new ``shortcut'' reduction rule $r$
  such that $t_1 \overset{r}{\mapsto} t_3$.
Of course, the trivial rule $r = t_1 \mapsto t_3$
  has this property,
  but only applies when reducing a term identical to $t_1$.
We seek instead a more general rule,
  which we call the composition
  of $r_1$ and $r_2$ (written $r_1 \circ r_2$),
  which depends only on $r_1$ and $r_2$
  and not on the specific terms $t_1$, $t_2$, and $t_3$.
This more general rule may apply to many future program states,
  not just $t_1$,
  and will accelerate their execution by combining two steps into one. 

We define rule composition by the following inference rule. 


\[
\inferrule{
        (\sigma_1, \sigma_2)\text{ is a most general unifier for }(p_2, p_3)\\
    }{
        (p_1 \mapsto p_2) \circ (p_3 \mapsto p_4) = \sigma_1[p_1] \mapsto \sigma_2[p_4] 
    }{}
\]

The use of the equal sign in this judgment form is not obviously justified, since $\sigma_1[p_1] \mapsto \sigma_2[p_4]$ may not be uniquely determined by $p_1 \mapsto p_2$ and $p_3 \mapsto p_4$. To justify our notation and our treatment of rule composition as a (partial) function, we have proven that rule composition respects an equivalence relation on rules. We provide only a sketch of the proof for this theorem, as for all theorems in this paper, with the full mechanized proof available in the accompanying Agda artifact. 

\begin{definition}
    Two rules are \textit{equivalent} iff they are the same up to renaming of variables. Formally, $p_1\mapsto p_2 \equiv p_1'\mapsto p_2'$ iff there exist $\sigma_1$ and $\sigma_2$ such that $\sigma_1[p_1] = p_1'$, $\sigma_1[p_2] = p_2'$, $\sigma_2[p_1'] = p_1$, and $\sigma_2[p_2'] = p_2$. 
\end{definition}

\begin{theorem}[Compatibility of Rule Composition]
    If $r_1\equiv r_1'$, $r_2\equiv r_2'$, $r_1\circ r_2 = r$, and $r_1'\circ r_2'\equiv r'$, then $r\equiv r'$.
\end{theorem}

\newcommand{\interp}[1]{\llbracket #1\rrbracket}

\begin{proof}[Proof Sketch]
    The theorem is proven by first interpreting each rule $r$ as a partial function $\interp{r}$ from patterns to patterns, then proving various properties of this interpretation: that $\interp{r_1} = \interp{r_2}$ iff $r_1\equiv r_2$ and that if $r_1\circ r_2 = r$, then $\interp{r_1}\circ \interp{r_2} = \interp{r}$. These properties can be combined to show that since partial function composition is unique, rule composition is unique up to equivalence. The same idea yields a proof that since partial function composition is associative, rule composition is associative up to equivalence. 
\end{proof}

We have proven two additional properties of rule composition
  that are important for \name:
  that the relevant rule compositions exist,
  and that the reduction ruleset can be
  safely extended to include them.
The first is~\Cref{theorem:existence},
  which guarantees that compositions exist
  for consecutive rules in a reduction sequence. The proof of this theorem is constructive, so the proof also provides an algorithm for computing the composition of rules. 

\begin{theorem}[Shortcut Existence]
    \label{theorem:existence}
    If $t_1 \overset{r_1}{\mapsto} t_2 \overset{r_2}{\mapsto} t_3$, then $r_1\circ r_2$ exists.    
\end{theorem}

\begin{proof}[Proof Sketch]
    The core of the proof is the lemma that if a pair of patterns $(p_1,p_2)$ has a unifier, then it has a most general unifier. This most general unifier is obtained by iteratively building up a pair of substitutions $(\sigma_1, \sigma_2)$ that any unifier must factor through, until this pair is also a unifier. $\sigma_1$ and $\sigma_2$ are initially both the identity substitution. At each step, $\sigma_1[p_1]$ and $\sigma_2[p_2]$ are compared. If they are equal, we are done. Otherwise, find a location where they disagree. If this location has a pattern variable $x$ in one of the two patterns (say $\sigma_1[p_1]$, without loss of generality) and a constructor term $k~\overline{p}$ in the other, then we update our substitution pair by extending $\sigma_1$ with a substitution mapping $x$ to $k~\{x_i\}_{i<n}$ for distinct fresh variables $x_i$, where $n$ is the length of $\overline{p}$. Every unifier must factor through these new substitutions because every unifier must assign $x$ to some pattern with the same constructor $k$ and the same number of children $n$. From here the unification process continues to iterate. This makes progress towards termination because the number of constructors in $\sigma_1[p_1]$ or $\sigma_2[p_2]$ increases at each step, but each cannot exceed the number of constructors in the original unifying pattern because every unifier factors through $(\sigma_1, \sigma_2)$. 

    Now consider the case where there is no point of disagreement between $\sigma_1[p_1]$ and $\sigma_2[p_2]$ of the aforementioned form: a variable corresponding to a constructor. They cannot disagree with two different constructors or different numbers of children at corresponding points, because then they could not have a unifier. We therefore have reduced the problem to the case where the two patterns disagree only at variable occurrences. We construct our most general unifier using the inductively defined notion of ``unified variables.'' Two variables $x_1$ and $x_2$ are unified if either $x_1 = x_2$, or if they appear in, say, $\sigma_1[p_1]$ in positions corresponding to variables $y_1$ and $y_2$ in $\sigma_2[p_2]$ and $y_1$ and $y_2$ are unified. Our most general unifier is the pair $(\sigma_1', \sigma_2')$ such that each substitution assigns unified variables to the same variable and non-unified variables to distinct variables. In other words, each maps injectively out of equivalence classes of variables under the ``unified'' relation. The second constraint defining the $\sigma_1'$ and $\sigma_2'$ is that if $x$ appears in $\sigma_1[p_1]$ in a position corresponding to variable $y$ in $\sigma_2[p_2]$, then $\sigma_1'(x) = \sigma_2'(y)$. The proof that this is the most general unifier follows from using the inductive notion of variable unification, yielding the finest partition of variables that would admit a unification of the patterns. 
\end{proof}


The second is~\Cref{theorem:validity},
  which states that the set of rules can be extended
  with compositions of its constituents
  without affecting program semantics.

\begin{definition}
    Rule sets $R$ and $R'$ are \textit{equivalent} iff they evaluate the same inputs to the same outputs. Formally, $R\equiv R'$ iff for all $t_1$ and $t_2$, $t_1 \overset{R}{\Rightarrow} t_2$ iff $t_1 \overset{R'}{\Rightarrow} t_2$.
\end{definition}

\begin{theorem}[Shortcut Validity]
    \label{theorem:validity}
    If $r_1, r_2\in R$ and $r_1\circ r_2$ exists,
    then $R \cup \{r_1\circ r_2\}\equiv R$.
\end{theorem}

\begin{proof}[Proof Sketch]
The core of the proof is the lemma that if $r_1\circ r_2$ justifies a step from $t_1$ to $t_2$, then $t_1$ can also step to $t_2$ by two steps: the first justified by $r_1$ and the second by $r_2$. This ensures that any occurrence of $r_1\circ r_2$ in a reduction sequence can be replaced by reductions only involving rules in $R$, and conversely if a term cannot step in $R$ then it also cannot step using the $r_1\circ r_2$. 
\end{proof}

Together, these theorems ensure that shortcut rules across reduction sequences always exist, and that \name may incorporate them into evaluation without affecting the result of the computation. We now turn to the description of the data structures and algorithms \name uses to find and apply these shortcuts. 

\section{Design and Implementation}
\subsection{The \name language}
Although our theory of shortcuts is generic
  over languages and reduction rule sets, 
  we instantiate the \name language as a functional, 
  first order programming language
  with Hindley-Milner types. 
The formal syntax of \name includes the following:

\[\begin{array}{lcllll}
x & \in & \Variable \\
C & \in & \Constructor \\
D & \in & \DataType \\
F & \in & \Function \\
i & \in & \Int \\
t & \in & \Type & \Coloneqq & \Int \mid D~\overline{t} \\
v & \in & \Value & \Coloneqq & i \mid C~\overline{v} \\
c & \in & \Control & \Coloneqq & x \mid \Let~x~=~c~\In~c \mid i \mid c~+~c \mid 
C~\overline{c} \mid \Match~c~\overline{[C~\overline{x}\rightarrow c]} \mid
F~\overline{c} \\
e & \in & \Env & \Coloneqq & \overline{x:=v} \\
k & \in & \Kont & \Coloneqq & \OK \mid \Let~x~=~\_~\In~c~k \mid C~\overline{v}~\_~e~\overline{c}~k \mid \Match~\_~\overline{[C~\overline{x}\rightarrow c]}~e~k \mid \\ 
& & & & \Add~\_~c~e~k \mid \Add~v~\_~k \mid F~\overline{v}~\_~e~\overline{c}~k  \\
d & \in & \Definition & \Coloneqq & D~\overline{x} = \overline{C~\overline{t}} \mid F~\overline{x} = c \\
s & \in & \State & \Coloneqq & \Eval~c~e~k \mid \Apply~v~k
\end{array}\]

This language is a generalization of the abstract machine language described in~\Cref{sec:abstract-machine}, allowing the user to define new data types and functions.
The semantics is straightforward:

\[\begin{array}{lcl}
    \Eval~(\Let~x~=~c_1~\In~c_2)~e~k & \mapsto & \Eval~c_1~e~(\Let~x~=~\_~\In~c_2~e~k) \\
    \Apply~v~(\Let~x~=~\In~c_1~e~k) & \mapsto & \Eval~c_1~[x~:=~v]~e~k \\
    \Eval~x~[y := v; e]~k & \mapsto & \Eval~x~e~k \\
    \Eval~x~[x := v; e]~k & \mapsto & \Apply~v~k \\[6pt]

    \Eval~C~e~k & \mapsto & \Apply~C~k \\
    \Eval~(C~c_1~\overline{c_2})~e~k & \mapsto &
    \Eval~c_1~e~(C~\_~e~\overline{c_2}~k) \\
    \Apply~v_2~(C~\overline{v_1}~\_~e~c_1~\overline{c_2}~k) & \mapsto &
    \Eval~c_1~e~(C~\overline{v_1}~v_2~\_~e~\overline{c_2}~k) \\
    \Apply~v_2~(C~\overline{v_1}~\_~e~k) & \mapsto & 
    \Apply~(C~\overline{v_1}~v_2)~k \\[6pt]

    \Eval~(\Match~c_1~\overline{[C~\overline{x}\rightarrow~c_2]})~e~k & \mapsto & 
    \Eval~c_1~e~(\Match~\_~\overline{[C~\overline{x}\rightarrow~c_2]}~e~k) \\ \Apply~(C_1~\overline{v_1})~(\Match~\_~[C_2~\cdots]~\overline{[C_3~\overline{x_3}\rightarrow~c_3]}~e~k) & \mapsto & \Apply~(C_1~\overline{v_1})~(\Match~\_~\overline{[C_3~\overline{x_3}\rightarrow~c_3]}~e~k) \\
    \Apply~(C_1~\overline{v_1})~(\Match~\_~[C_1~\overline{x_1} \rightarrow c_1]~\cdots~e~k) & \mapsto & \Eval~c_1~(\overline{v_1 := x_1};e)~k \\[6pt]

    \Eval~i~e~k & \mapsto & \Apply~i~k \\
    \Eval~(c_1 + c_2)~e~k & \mapsto & \Eval~c_1~e~(\Add~\_~c_2~e~k) \\
    \Apply~i_1~(\Add~\_~c_2~e~k) & \mapsto & \Eval~c_2~e~(\Add~i_1~\_~k) \\
    \Apply~i_2~(\Add~i_1~\_~k) & \mapsto & \Apply~(i_1 + i_2)~k \\[6pt]
    
    \Eval~F~e~k & \mapsto & \Eval~F.c~k \\
    \Eval~(F~c_1~\overline{c_2})~e~k & \mapsto & \Eval~c_1~e~(F~\_~\overline{c_2}~e~k) \\
    \Apply~v_2~(F~\overline{v_1}~\_~e~c_1~\overline{c_2}~k) & \mapsto &
    \Eval~c_1~e~(F~\overline{v_1}~v_2~\_~e~\overline{c_2}~k) \\
    \Apply~v_2~(F~\overline{v_1}~\_~e~k) & \mapsto & 
    \Eval~(F.c~\overline{v_1}~v_2)~e~k
\end{array}\]

The input to \name is a sequence of \Definition{}s, each of which defines an algebraic data type or a function.
The \name compiler then compiles these definitions to an OCaml interface file, which exports the corresponding type definitions, alongside the optimized function definitions on those types. These functions are executed using shortcut memoization, so repeated calls will obtain speedup from caching.

\subsection{Compilation}
Applying shortcut memoization directly to this semantics,
  however, would have unacceptable overhead due to excessive branching (every step inspects the form of the program state) and linear environment traversal for variable instantiation. The \name compiler applies various optimizations to address these performance obstacles. 

The compiled program represents the program state as a control-environment-continuation record,
  with $\Apply$ represented as
  a special control state with a single value in the environment.
The control is represented by a kind of program counter, an integer index into an array of steps, each of which is a function from program state to program state.
The environment is represented as a dynamic array of values that program variables index into, and the continuation is represented as a stack.
This substantially reduces the overhead of interpretation.

We also compact the program so that each execution step contains multiple basic instructions. Sequences of variable lookups, let definitions,
  constructor applications,
  and integer manipulations
  are combined into a single step, since these operations never allocate continuation frames. Instead of intermediate values being stored as continuations, they are pushed onto and popped from the environment.
Only calls to non-tail-recursive functions require pushing a frame onto the continuation stack. At these steps, unused variables are removed from the continuation~\cite{safe-for-space}, which reduces the number of variables in learned shortcuts and makes shortcuts faster to apply. 
Constructors and pattern matching in \name are compiled to the corresponding OCaml constructs, which provide more efficient matching than the linear scan implied by the \name language reduction rules. 
These optimizations not only decrease execution time, 
  but also make each primitive rewrite rule larger, decreasing the amount of rule manipulations needed. 

Our compiled representation is thus more like a stack or register machine than a pure CEK machine, and enables more efficient execution due to its resemblance to actual hardware.

\subsection{Running \name}

To execute a function in this compiled representation,
  \name sets up an abstract machine
  with the control pointing to the beginning of the function,
  the environment containing the function arguments,
  and the continuation initialized to a special \Done~value.
\name then runs in a loop until the abstract machine
  attempts to \Apply \ a value to the \Done \ continuation, 
  at which point \name can extract the result from the environment.
In the loop,
  \name first searches for a learned shortcut
  applicable to the current abstract machine state.
If one is found, it uses that shortcut
  to rewrite the abstract machine state;
  if one is not found, such as when a new arithmetic operation must be performed, it instead generates
  the applicable atomic rewrite rule
  and applies that.
These atomic rewrites are generated on demand
  to model non-rewrite-based operations
  like arithmetic on machine integers.
To generate atomic rewrites,
 \name executes the program instructions normally, but marks
  any environment examined or continuation slot as used.
\name then constructs an LHS and RHS pattern
  from the initial and final program state,
  using variables for any unmarked slots
  and, for marked slots, instantiating the pattern with
  constructors until it reaches an unmarked value.

\subsection{Shortcut Learning}
\label{sec:learn}

After applying an atomic or shortcut rewrite,
  \name may compose it with a previous rewrite
  to learn a new shortcut.
Since any sequence of successive rules can be composed,
  there are $O(n^2)$ possible shortcuts
  that can be derived from a length reduction sequence $n$.
Thus, generating all possible shortcuts would take
  longer than actually running the program.
\name is therefore more selective
  about which shortcuts to derive and keep,
  and uses various optimizations
  to store shortcuts compactly
  and search for applicable shortcuts quickly.

To prevent an explosion in the number of shortcuts,
  \name separates shortcuts into ``levels'',
  where rules of the same levels have
  arch of the same length in \Cref{fig:shortcut-before},
  and only compose adjacent pairs of the same level.
This limits the number of shortcuts learned
  to be linear in the length of the reduction sequence,
  while still learning some very large
  (and thus very beneficial) shortcuts.
Furthermore, in this set of shortcuts,
  no shortcut skips past the start of another.
This critically allows shortcut application to be greedy;
  applying a shortcut at the current program state
  is unlikely to skip over some intermediate state
  where a better, longer shortcut would be available.

\begin{figure}[htbp]
\begin{subfigure}{0.9\textwidth}
\centering
\begin{tikzpicture}[
  node/.style={draw, circle, minimum size=6mm},
  edge/.style={-Latex},
]

\node[node] (n0) at (0,0) {A};
\node[node] (n1) at (1,0) {B};
\node[node] (n2) at (2,0) {C};
\node[node] (n3) at (3,0) {D};
\node[node] (n4) at (4,0) {E};
\node[node] (n5) at (5,0) {F};
\node[node] (n6) at (6,0) {G};
\node[node] (n7) at (7,0) {H};

\draw[edge] (n0) -- node[above]{} (n1);
\draw[edge] (n1) -- node[above]{} (n2);
\draw[edge] (n2) -- node[above]{} (n3);
\draw[edge] (n3) -- node[above]{} (n4);
\draw[edge] (n4) -- node[above]{} (n5);
\draw[edge] (n5) -- node[above]{} (n6);
\draw[edge] (n6) -- node[above]{} (n7);

\draw[bend left=60] (n0) to (n2);
\draw[bend left=60] (n2) to (n4);
\draw[bend left=60] (n4) to (n6);
\draw[bend left=60] (n0) to (n4);

\end{tikzpicture}
\caption{
  The shortcuts obtained
    after executing seven rewrite rules.
  The horizontal arrows are
    atomic reduction steps,
    while the bowed arcs are
    shortcut rewrite rules obtained by rule composition. }
\Description{
  A shortcut hierarchy
    over a linear execution trace from A to H.
  Horizontal arrows denote
   atomic reduction steps.
  Curved arrows denote learned
    shortcuts formed by composing
    adjacent steps of the same level.}
\label{fig:shortcut-before}
\end{subfigure}

\begin{subfigure}{0.9\textwidth}
\centering
\begin{tikzpicture}[
  node/.style={draw, circle, minimum size=6mm},
  edge/.style={-Latex},
]

\node[node] (n0) at (0,0) {A};
\node[node] (n1) at (1,0) {B};
\node[node] (n2) at (2,0) {C};
\node[node] (n3) at (3,0) {D};
\node[node] (n4) at (4,0) {E};
\node[node] (n5) at (5,0) {F};
\node[node] (n6) at (6,0) {G};
\node[node] (n7) at (7,0) {H};
\node[node] (n8) at (8,0) {I};

\draw[edge] (n0) -- node[above]{} (n1);
\draw[edge] (n1) -- node[above]{} (n2);
\draw[edge] (n2) -- node[above]{} (n3);
\draw[edge] (n3) -- node[above]{} (n4);
\draw[edge] (n4) -- node[above]{} (n5);
\draw[edge] (n5) -- node[above]{} (n6);
\draw[edge] (n6) -- node[above]{} (n7);
\draw[edge] (n7) -- node[above]{} (n8);

\draw[bend left=60] (n0) to (n2);
\draw[bend left=60] (n2) to (n4);
\draw[bend left=60] (n4) to (n6);
\draw[bend left=60] (n6) to (n8);

\draw[bend left=60] (n0) to (n4);
\draw[bend left=60] (n4) to (n8);

\draw[bend left=60] (n0) to (n8);

\end{tikzpicture}
\caption{The shortcuts after the addition of another reduction step from H to I, because now the number of steps is a power of two. Three new shortcuts are added: one from G to I, one from E to I, and finally one from A to I.}
\Description{The same shortcut hierarchy after adding one more execution step
  from H to I. New shortcuts are added from G to I, from E to I, and from A to
  I, illustrating incremental maintenance of the hierarchy.}
\label{fig:shortcut-after}
\end{subfigure}
\end{figure}
The hierarchical structure of the shortcut tree
  makes it easy to maintain throughout evaluation, starting out empty and growing with each successive reduction step.
\Cref{fig:shortcut-after} illustrates how
  the shortcut set in~\Cref{fig:shortcut-before}
  grows with the addition of another reduction step.
In the figure,
  the atomic reduction step from H to I
  is composed with its preceding step from G to H,
  forming a shortcut from G to I.
Since there is a preceding shortcut at the same level,
  from E to G, these two are composed
  to form a shortcut from E to I,
  and the process repeats again to form
  an even higher-level shortcut from A to I.
Shortcuts are added continuously during evaluation
  and immediately become available for use.
This means that a shortcut derived
  by composing rules at the beginning of an evaluation
  can be used as a single step further along in the evaluation.
This can accelerate even a single execution of a program
  if that execution repeatedly executes
  similar programs on similar data,
  such as during recursion.

\subsection{Discrimination Tree}
\label{sec:lookup}

This set of learned shortcuts must be indexed
  for efficient search.
More precisely, at each execution step,
  \name must find a reduction rule
  applicable to the current program state.
There may be multiple such rules, in which case \name prefers to apply the rule with the greatest ``length''. The length of a rule is the number of atomic reduction steps it skips. This quantity is stored alongside each rule and is easily maintained via $\text{length}(r_1\circ r_2) = \text{length}(r_1) + \text{length}(r_2)$. 

The set of available reduction rules, including both atomic rules and learned shortcuts, is maintained in a \textit{discrimination tree}, a data structure commonly used in theorem provers. The discrimination tree is a trie storing the left-hand sides (LHS) of all known shortcuts or atomic steps. Since a trie's keys are sequences, each LHS is flattened into its preorder traversal. Additionally, since reduction rule lookup should be invariant under alpha equivalence (variable renaming), reduction rules are represented using De Bruijn indices. Each variable on the LHS is a nameless binder (a special nullary constructor), and each variable on the RHS is a number indexing which LHS binder it corresponds to. This approach is only possible if the rewrite rules are \textit{linear} in the sense of having no repeated variables on the LHS, which is true of the CEK abstract machine. 

In the trie data structure, each key is a complete path through the tree, and each key's value is stored at that path's leaf. Our trie is a \textit{radix tree}, 
  meaning that each node stores a sequence of elements rather than a single element. 
The sequence at each node is as long as possible,
  calculated using a longest common prefix algorithm. The consequence of this is that at a given node, every child's sequence begins with a different element, which the parent uses to index that child.
In \name, the children are stored in
  an association list or hash table,
  depending on their number.%
\footnote{
When the number of children does not exceed the threshold, the overhead of the hashtable is not worth it.
}
An example is shown in~\Cref{fig:discrimination-tree}.

\name finds the longest reduction rule
  that matches the current program state by flattening the program state into its preorder traversal and matching this sequence into the trie. Matching proceeds by traversing the trie from the root, sequentially matching nodes in the trie with segments in the flattened program state. When a variable is encountered in the trie, it matches the prefix of the flattened program state corresponding to a complete subterm.
This traversal is branching, since the program state can match both a constructor in the trie and a variable in the tree, leading to different children. The branching traversal reaches some set of leaves, and the rule with the highest length among these leaves is returned.

\begin{figure}[t]
\centering
\includegraphics[width=0.5\linewidth]{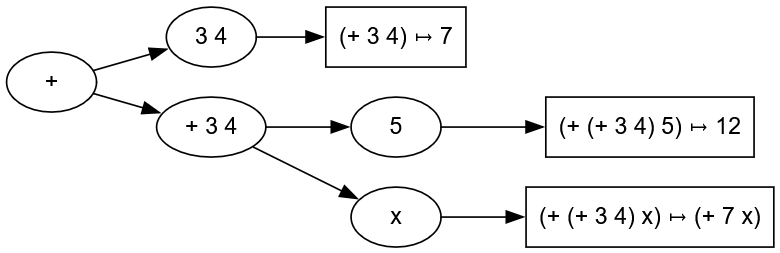}
\caption{A discrimination tree holding three rewrite rules.} 
\Description{A discrimination tree storing three rewrite rules. Internal
nodes contain constructor prefixes, and leaves contain complete rewrite
rules. The example shows how common prefixes are shared across rules.}
\label{fig:discrimination-tree}
\end{figure}
\section{Optimization}
\label{sec:opt}
Several optimizations are necessary
  to ensure that shortcuts can be
  looked up and applied quickly enough
  to actually speed up program execution.

\subsection{Search Pruning}
Searching the discrimination tree is expensive
  because there may be many shortcuts
  that match a given program state;
  in fact, our shortcut learning strategy essentially guarantees this.
To reduce search time,
  we use a branch-and-bound algorithm.
Every node in the discrimination tree
  stores the length of the longest rule in its subtree.
During traversal of the discrimination tree, \name maintains the length of the longest rule found so far in the search and skips subtrees without a greater maximum rule length.
Additionally, branches with a longer maximum rule length are searched first in an attempt to maximize pruning.

\subsection{Fast Rule Matching}
\label{subsec:data-rep}
As more complex shortcuts are learned, 
  the rules become correspondingly more complex. 
In a standard data representation for terms and rules,
  the cost of applying a rule is proportional
  to the size of the LHS and RHS.
A syntactically large shortcut might then provide
  no meaningful speedup over direct computation.
\name therefore uses an alternative encoding of rewrite rules
  in which pattern matching is dominated instead by the number of variables.

Recall that the discrimination tree stores
  the prefix traversals of the LHSs of learned shortcuts.
To search it efficiently,
  the program state is also stored as a string,
  its pre-order traversal;
  we call the resulting string the ``program state string''.
Searching the discrimination tree
  for the program state string requires only
  testing that the program state starts with a fixed string
  to match constant strings,
  and finding full subterms in the program state string
  to match variables.
To accelerate these two steps,
  the program state string is stored in
  an augmented balanced binary tree.%
\footnote{
  Constant strings in the discrimination tree
    use the same balanced binary tree representation
    to speed up common prefix checks needed
    to update the discrimination tree.
}
\footnote{
  This tree is implemented as an optimized finger tree,
    which generally provides a better constant
    and algorithmic bound.
}
Both substring matching and full term matching
  take $O(\log n)$ time on this binary tree,
  which makes applying a rule with $k$ variables
  take $O(k\log n)$ time.
The augmentations accelerate string and subterm matching
  using monoidal parsing and monoidal hashing.
Applying rules also involves constructing new nodes,
  which requires concatenating
  subterms and constant strings for each shortcut's RHS.
Balanced binary trees also support
  $O(\log n)$ time concatenation,
  so applying rules is also fast,
  with a time complexity dependent on variable count
  but not rule size.%
\footnote{
  Each variable ends up bound to
    a substring of the program state string
    containing $O(\log n)$ subtrees,
    so applying a rule requires concatenating
    $O(k \log n)$ subtrees,
    which takes $O(k \log n \log \log n)$ time.
    What is important is not the exact asymptotics,
    but the independence of rule size.
}
Inserting new rules to the discrimination tree
  is likewise sped up.

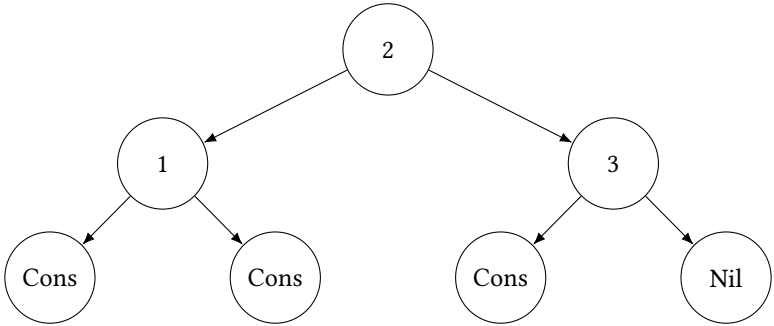
\begin{figure}[t]
\centering
\begin{tikzpicture}[
  node/.style={draw, circle, minimum size=12mm},
  edge/.style={-Latex}
]

\node[node] (n0) at (0,0) {2};

\node[node] (n1) at (-3,-1.5) {1};
\node[node] (n2) at ( 3,-1.5) {3};

\node[node] (n3) at (-4.5,-3) {Cons};
\node[node] (n4) at (-1.5,-3) {Cons};
\node[node] (n5) at ( 1.5,-3) {Cons};
\node[node] (n6) at ( 4.5,-3) {Nil};

\draw[edge] (n0) -- (n1);
\draw[edge] (n0) -- (n2);

\draw[edge] (n1) -- (n3);
\draw[edge] (n1) -- (n4);

\draw[edge] (n2) -- (n5);
\draw[edge] (n2) -- (n6);

\end{tikzpicture}
\caption{\name's internal term representation, which allows efficient wrapping and unwrapping of multiple constructors in a combined operation. In this example, wrapping the term in $Cons(0,\_)$ prepends two constructor nodes ($Cons$ and $0$), and the binary tree structure provides a logarithmic time append operation.}
\Description{A balanced binary tree representation of a term string. Internal
  nodes summarize concatenated subsequences, and leaves store constructors such
  as Cons and Nil. The figure illustrates efficient wrapping and concatenation
  operations.}
\end{figure}

\paragraph{Monoidal Parsing}
Monoidal parsing can quickly find a full subterm
  beginning at a given index in the program state string.
In monoidal parsing,
  each constructor of arity $a$ is assigned a degree $1 - a$.
This quantity represents
  the constructor's contribution to
  the total number of subterms in a sequence;
  it contributes $+1$ for the subterm it constructs
  and contributes $-a$ for the $a$ children it requires.
For example, every nullary constructor has degree $1$
  and contributes its own subterm,
  while every ternary constructor has degree $-2$
  and combines three subterms into one.
The complete subterm in $S$ that begins at a given index $i$
  is then the smallest $j$ such that
  the summing of degrees in $S[i:j]$ produces $1$,
  which corresponds to exactly one subterm.
Then $j$ can be found efficiently using
  the binary tree structure.
Each binary tree node is augmented with two values:
  the sum of degrees in the sequence it stores ($\TotalDegree$)
  and the maximum $\TotalDegree$ across
  all prefixes of the sequences it stores, $\MaxDegree$.
Each prefix of the sequence corresponds to a forest of subtrees 
  in the binary tree. 
The $\TotalDegree$ and $\MaxDegree$ in each root of this forest 
  can be used to calculate the $\MaxDegree$ of the entire prefix. 
Since $\MaxDegree$ does not decrease across prefixes in the sequence, 
  this method can be used to find the shortest prefix 
  with $\MaxDegree = 1$ by binary search.

\paragraph{Monoidal Hashing}
A monoidal hash function is a hash function on sequences $h$
  with a concatenation operator $\oplus$ on hash values
  such that $h(x {+\!\!+} y) = h(x) \oplus h(y)$.
This means that the hash of a sequence
  can be built by concatenating the hashes of the elements. 
Every node in the binary tree stores
  its length and the monoidal hash of its contents.
To match constant strings,
  length information is first used
  to reduce the constant string to
  binary tree nodes $O(\log n)$,
  which are monoidally hashed together and compared
  to the monoidal hash of the constant string.

\name uses the Castagnoli CRC (CRC32C)
  for its monoidal hash,
  which is implemented in x86 SSE4.2.
The monoidal hash concatenation operation
  is implemented by storing, for each node,
  both the checksum $p$
  and the length-dependent factor $x^n \bmod P$.
The checksum of the composition
  is then just the combination of the two checksums
  and the length-dependent factor.
Our implementation leverages hardware support:
  hardware CRC instructions for byte-wise folding
  and carry-less multiplication
  (\texttt{PCLMULQDQ} on x86\_64, \texttt{PMULL} on AArch64)
  for monoidal hash combination.

\subsection{Instantiation}

Because the cost of matching and applying shortcuts
  is now determined by the number of pattern variables,
  \name prefers to learn shortcuts with fewer variables.
This is achieved by
  limiting the shortcuts to a maximum of ten pattern variables
  in the continuation and in each entry in the environment.
If unification results in more pattern variables for the composed rule,
  the extra variables are instead instantiated with
  their values in the current program state.
More specifically, in both the continuation and in each environment entry,
  the rightmost ten variables are kept,
  while the remainder are instantiated.
The rightmost ten variables are preferred because
  many common data structures (like lists)
  are imbalanced to the right,
  and because function applications,
  which often involve constant functions,
  are on the left.
\section{Evaluation}
\newcommand{\hazelListLength}{400}
\label{sec:eval}
To evaluate \name, 
  we devised two benchmarks to test different aspects of the system: 
\begin{itemize}
    \item[RQ1.] Can \name accelerate evaluation across program changes?
    \item[RQ2.] Can \name accelerate evaluation with repetitive patterns of computation?
\end{itemize}

\subsection{RQ1: Can \name accelerate evaluation across program changes?}

Hazel is a purely functional, call by value programming language
  that supports live program through its ``hole'' construct.
The Hazel system allows users to edit programs, 
  and re-evaluates the partial program on every keystroke,
  showing users live results as they type.
We implemented a (untyped) Hazel interpreter in \name
  and used it to evaluate traces, captured by Hazel,
  of three students completing functional programming exercises.
Each trace contains the entire program after every keystroke;
  we execute all programs in the trace, just as Hazel would,
  while retaining the learned shortcuts.
We ignore duplicate programs,
  which \name would otherwise speed up enormously.
We compare the evaluation time between \name
  and a variant of \name that does not
  learn or apply shortcuts.

The benchmark consisted of eight programs of increasing size and complexity, which we asked our participants to implement based on a specification. 
The programs were:
  \begin{itemize}
      \item Basic list functions (\texttt{map}, \texttt{filter}, \texttt{append}).
      \item \texttt{pair}, which takes a list as input and returns the list of adjacent pairs in the input list, requiring non-structural recursion.
      \item \texttt{reverse}, which is implemented with an \texttt{append} helper function.
      \item Functions for quick sort, merge sort, and insertion sort.
  \end{itemize}

The 24~total benchmark traces contain
  a total of \hazelPointCount~recorded programs.
Each experiment was run in a fresh cache
  on Ubuntu 24.04.4 with an Intel i7-8700K CPU
  and 32 GiB of memory.

The overall speedup across all traces is \hazelTotalSpeedup.
This speed-up is significantly affected
  by the size of the input list used.
All results in this section runs programs
  on input lists with \hazelListLength~elements;
  since \name provides asymptotic speedups
  for repeated computations,
  longer lists would provide even bigger speedups.
We view \name's ability to extract speedups
  even for relatively short lists
  of a few hundred elements as a strength.

\Cref{fig:benchmark} shows individual evaluation times
  as a scatterplot and \Cref{fig:cdf} shows
  the cumulative distribution function of speedups
  for each data point in the benchmark.
Each point's horizontal position in \Cref{fig:benchmark} is its evaluation time
  without learned shortcuts,
  and each point's vertical position is its evaluation time
  with learned shortcuts.
The heavy diagonal line plots the $y = x$ equal-time curve;
  points below it are cases where incremental evaluation is faster.
Most points are below the line, often substantially---%
  note that the plot uses a log scale---%
  indicating that incremental evaluation often provides a large speedup.

The scatterplot contains many vertical bands,
  as well as a single horizontal band. 
Vertical bands arise from cases
  where the computation is essentially unchanged between edits,
  resulting in an approximately constant baseline time
  but more efficient incremental computation
  as shortcuts are learned across multiple executions.
The horizontal band arises from
  repetitive computation in which memoization renders
  the incremental time insensitive to variations in programs.
  
An aggregate view of the data,
  focusing purely on the speedup over baseline,
  is shown in \Cref{fig:cdf}. It show three significant bands:
  one at a roughly $2\times$ slowdown,
  affecting roughly $6\%$ of programs,
  and another at a roughly $25\times$ speedup,
  affecting roughly $70\%$ of them.
Focusing on slowdowns,
  these primarily involve recursive computations,
  where \name spends time recording shortcuts
  that are not applicable in future iterations due to changed inputs.
The largest slowdowns are roughly $8\times$.
In these cases,
  the user has edited the program in such a way
  that requires wholly new computations never performed before,
  for example, when a function is called for the first time. 
In such cases, \name pays an upfront cost:
  it spends time learning new shortcuts
  that will not be applied until later programs in the trace.
The overhead of this process
  is responsible for the slowdown on these points.
Fortunately, most of these slowdowns
  appear at the beginning of the trace,
  when programs are largely incomplete and run quickly.

The net speedups for each interpreted function
  are recorded in \Cref{tab:individual}.
Speedups vary depending on the user 
  due to coding style differences.
For example, users differed in whether
  they implemented the base case or the inductive case first,
  and in what order they placed the main function
  and the helper functions.
Two of Cam's\footnote{Not a real name} programs included infinite loops,
  so they could not be assigned speedup factors.%
\footnote{
An interesting side effect of shortcuts is that
  they can detect some infinite loops:
  a shortcut LHS $\mapsto$ RHS where the LHS matches the RHS
  guarantees non-termination.
}
We also measured
  the proportional time spent by each component of \name,
  recorded in \Cref{tab:component}.
Rule lookup and rule composition take the most time
, since the former requires branching
 and the latter executes complex unification logic.
These results justify the optimizations used
  to speed up lookup.

\Cref{tab:individual} also shows
  memory overhead for each function. 
The overhead is quite high,
  sometimes even reaching $100\times$.
More complex functions have a higher overhead,
  as the evaluation is longer, and thus more shortcuts are learned.
The number is high mainly because
  no cache eviction policy is currently implemented in \name.
Like in any cache,
  rules in discrimination tree
  have wildly different hit counts,
  and we expect even simple cache eviction policies
  to dramatically reduce memory overhead.

\begin{figure}[htbp]
\begin{subfigure}[t]{0.45\textwidth}
    \centering
    \includegraphics[width=1.0\textwidth]{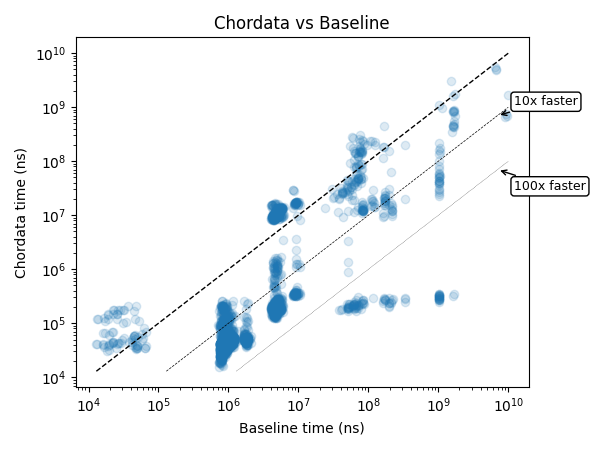}
    \caption{Scatterplot of incremental vs. baseline evaluation time.}
    \label{fig:benchmark}
\end{subfigure}
\hfill
\begin{subfigure}[t]{0.45\textwidth}
    \centering
    \includegraphics[width=1.0\textwidth]{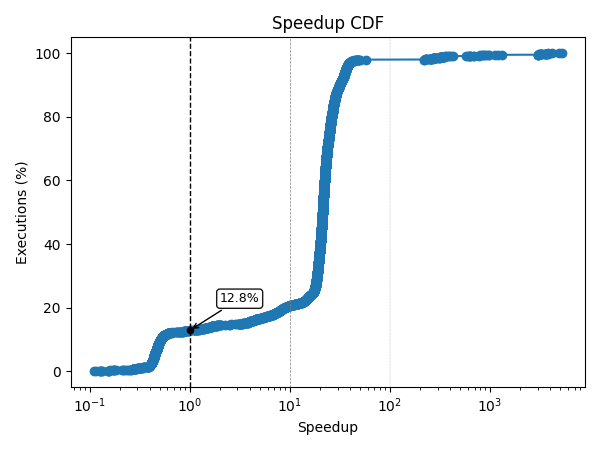}
    \caption{CDF of per-state speedups.}
    \label{fig:cdf}
\end{subfigure}
\caption{Results for the live-programming benchmark.}
\Description{Two plots for the live-programming benchmark. The left plot is a
  scatterplot comparing incremental evaluation time with baseline evaluation
  time for each program state. The right plot is a cumulative distribution of
  observed speedups.}
\end{figure}  

\begin{table}[htbp]
\centering
\begin{subtable}[t]{0.70\linewidth}
  \centering
  \hazelSpeedupTable
  \caption{Speedup and memory overhead by benchmark.}
  \label{tab:individual}
\end{subtable}
\begin{subtable}[t]{0.20\linewidth}
  \centering
  \hazelSpeedBreakdown
  \caption{Time breakdown by component.}
  \label{tab:component}
\end{subtable}
\caption{Detailed results for the live-programming benchmark.}
\label{tab:hazel-results}
\end{table}

\subsection{RQ2: Can \name accelerate evaluation with repetitive patterns of computation?}
The Hazel evaluation covers our
  expected use case for \name:
  live programming where similar programs are
  repeatedly evaluated and latency is critical.
However, during the course of that evaluation,
  we noticed that \name would sometimes speed up
  a single program evaluation with a blank cache.
This is because the program might itself generate
  repetition computation,
  such as by transforming an input into
  a more repetitive form.
(For example, merge sort produces
  many sorted lists of similar sizes!)
To cleanly demonstrate this phenomenon,
  we performed a second case study
  with a symbolic differentiator and an algebraic simplifier,
  which differentiates a mathematical expression twice
  and then simplifies the result.
The simplification procedure
  reduces additions and multiplications
  with their respective identities,
  combines constants appearing in
  the same sequence of additions or multiplications,
  and factors out common factors within sequences of additions. 
This process is iterated until a fixed point is reached.
Symbolic differentiation frequently
  results in similar subterms with repeated structure
  (repetition across space),
  whose differentiation and simplification
  can then be accelerated with shortcuts.
The iterative nature of simplification, meanwhile,
  creates repeated structures across iterations
  (repetition across time).
Since each simplification changes
  only a part of the expression,
  processing the unchanged parts of the expression
  is sped up by shortcuts learned during the previous iteration.

\begin{figure}
    \centering
    \includegraphics[width=0.5\linewidth]{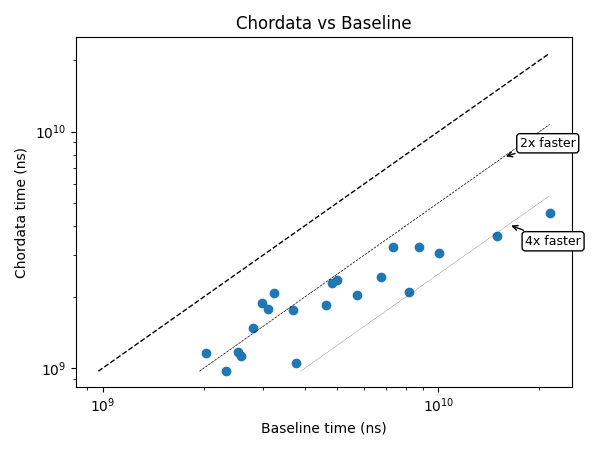}
    \caption{The scatterplot for the algebraic simplification benchmark}
    \Description{A scatterplot for the algebraic simplification benchmark
  comparing incremental execution time with baseline execution time across
  randomly generated inputs.}
    \label{fig:arith-scatter}
\end{figure}

Using these two kinds of repetition, 
  \name is able to obtain a \arithTotalSpeedup{} speedup,
  on randomly generated input expression,
  sized between 100 and 200 nodes.
To maximally eliminate shared structure,
  the cache is emptied between each runs.
Although this speedup is much smaller than 
  that obtained in our Hazel benchmark,
  it is still significant:
  the Hazel benchmark evaluated many similar programs
  with a \emph{shared} cache,
  while this experiment does not reuse
  caches across experiments.
Although applying \name to general computation
  will likely cause unacceptable memory overheads,
  we are excited about its ability
  to exploit repetition in input data,
  in program execution, and across edits. 
Furthermore,
  the speedup, while small,
  is obtained despite the overhead of shortcut lookup,
  and the speedup would be even bigger
  with larger or more repetitive expressions.
For example, an application that performs
  repeated differentiation,
  such as a Taylor expander,
  would create more shared structure,
  and thus achieve greater speedups.
\section{Related Work}
\paragraph{Self-Adjusting Computation}
Self-adjusting computation \cite{SAC} and Adapton \cite{adapton}
  are two prominent incremental computing systems
  that record dependencies between values
  in a computation graph.
They then use graph algorithms to recompute
  a minimal subset of the graph
  when inputs to a computation change.
However, changes to the computation graph
  are difficult for such systems to process.
Some later systems are hybrid~\citet{Acar2003AdaptiveM}, 
  relying on tight interplay between memoization and computation graphs to
  maintain identities and register changes;
  for example, ``nominal Adopton''~\cite{nominal-adapton}
  uses user-directed \emph{names}
  as stable identifiers in the computation graph.
Moreover, the design of the computation graph
  often restricts what incremental computations are supported.
For example, self-adjusting computation
  does not naturally support incrementalization
  across reordered or shared subcomputations,
  as it maintains a global linear order on computations
  using an order-maintenance data structure. 
In Adapton, a related limitation
  arises from the phase separation
  between dirtying nodes and recomputing values.
This design can lead to unnecessary recomputations
  that a more demand-driven evaluation strategy would avoid.

\paragraph{Incremental Lambda Calculus}
Other works like 
  the incremental lambda calculus~\cite{incremental-lc},
  patch Reconciliation~\cite{patch-reconciliation},
  and differential dataflow~\cite{differential-dataflow}
  propagate changes to program values along a dependency graph.
For example,
  in differential dataflow~\cite{differential-dataflow}
  program values are bags,
  and changes are additions and deletions of bag items.
Incremental lambda calculus~\cite{incremental-lc}
  provides a unifying framework
  to track such changes
  on arbitrary lambda calculus programs.
However,
  like the computational graph approaches,
  it only responds to changes
  and thus also suffers from the distal similarity problem.

\paragraph{Memoization}
Memoization speeds up function execution
  by recording the input/output pairs of the function
  and reusing the previous output 
  if it matches the corresponding previous input,
  thereby skipping the computation.
A problem with memoization is its limited generality:
  recording every previous input to a memoized computation
  does not help if new values differ slightly.
Shortcut memoization is a type of memoization, 
  but improves the performance and generality 
  of traditional memoization by learning
  generalizable rules instead of simple input-output pairs.
Dependency tracking schemes~\cite{precise-dependencies, caching-dependencies}
  similarly improve memoization
  by identifying portions of the input
  that the memoized function does not depend on.
Shortcut memoization's use of unification
  is in some sense this sort of dependency tracking scheme.

Another problem for memoization,
  especially on recursive functions
  is that naive memoization must traverse the spine of the input
  to locate and resolve changes, which can be slow.
Approaches like function caching~\cite{function-caching}
  introduce ``chunky decomposition'',
  an algorithm to convert arbitrary lists
  into balanced binary trees,
  with similar lists sharing a tree structure.
This reduces the length of the spine
  and improves memoization performance.
\name uses a similar idea
  to speed up lookup in its discrimination tree,
  but generalizes it beyond lists
  using monoidal hashing and parsing.
  
\paragraph{Partial Evaluation}
Partial evaluation and staging are compilation techniques
  that partition function inputs into those
  known at compile time and those known at run time.
The function is then evaluated as much as possible,
 utilizing the compile time portion,
  generating a residual program that evaluates the runtime input.
\citet{Sundaresh1992IncrementalCV} 
 explores the relationship between
 partial evaluation and incremental computing
 by building an algorithm that incrementalizes arbitrary programs,
 with the user specifying which inputs are fixed at compile time
 and which can vary during runtime.
\name shares some similarities with partial evaluation,
  shortcut LHSs having both fixed inputs and pattern variables.
A key difference is that typical partial evaluators
``push'' while memoization typically ``pulls''.
In partial evaluation, that is,  some inputs are specified,
 and the function is reduced as far as possible at compile time,
 with a residual program generated for a separate runtime stage.
In \name, shortcuts are learned and composed,
  with the input dependencies discovered via unification,
  at run time.

\paragraph{Incremental Program Analysis}
Incremental techniques present substantial opportunities in program analysis, with a range of potential applications. One prominent application is integration with integrated development environments (IDEs), where analysis results must be updated responsively following each program modification to notify the user of potential issues in the current program.

Prior work on incremental program analysis has been developed within several distinct analysis frameworks, including fixpoint iteration~\cite{ryder1983incremental,arzt2014reviser,PollockS89, nougrahiya2025incidfa}, Datalog-based analysis~\cite{szabo2016inca,szabo2021incremental,ZhaoSRS21}, and algebraic program analysis~\cite{zhou2025incremental}. However, because these approaches rely on different underlying analysis frameworks, the resulting techniques tend to be heterogeneous and tightly coupled to the specific framework in which they are formulated. Besides the heterogeneity in frameworks, the sources of increments are also diverse: (1). The input program might change, and the IDE might require responsive analysis results, so the input could be incrementalized. (2). During fixpoint computation typical in program analysis, it is possible to record what has changed and propagate only those changes instead of recomputing everything~\cite{Quiring2024DerivingWD}. (3). Authors like \citet{Zhu2004SymbolicPA} have pointed out that different variables in the program might share many analysis results in common, and the similarity can be exploited by methods such as BDDs~\cite{NaikAW06}. \name doesn't specifically focus on program analysis, but might be applicable to the same class of problems.

\paragraph{Hazel}
Hazel has many subsystems all of which face similar
  latency constraints imposed by interactivity and liveness,
  and thus could be incrementalized for better performance.
Such subsystems include typechecking~\cite{hazel-OM} and 
  syntax error correction~\cite{hazel-EC},
  and \name might be applicable to such subsystems
  to improve incrementality. 
  Recent work on incremental typechecking in Hazel approached the problem using a dependency graph based approach~\cite{hazel-OM},
    but it may be the case that a hybrid approach that propagates along the dependency graph but also stores and recalls from a cache could be more effective when programs are being edited in the manner we observed. 

\paragraph{Succinct Data Structures}
Succinct data structures represent repetitive data compactly;
 they generally rely on splitting the input into chunks,
 storing computation results in a lookup table,
 and reusing those results, somewhat similar to \name.
Moreover, \name's data representation 
  for fast rule matching and application
  borrow ideas from succinct representation,
  especially succinct trees
  encoding a preorder traversal with degree information
  into a string.
Authors like \citet{Gawrychowski2015OptimalDS},
  \citet{Sadakane2007SuccinctDS}, and
 \citet{He2010SuccinctRO} 
  have explored similar string representations.
For example,
  \cite{Gawrychowski2015OptimalDS}
 encodes a family of strings via
  straight line grammars,
  which is an acyclic context free grammar
  where each nonterminal has a single production rule,
 and thus each nonterminal represents exactly one string.
This encoding can thus compress and deduplicate identical substrings,
 reducing the total memory consumption of \name,
  as well as removing the probabilistic reliance on hashing.
We hope to explore succinct data structure ideas further in \name.

\paragraph{Term Indexing}
Discrimination trees \cite{McCuneWilliam} belong to a family
  of data structures and algorithms used
  in automated theorem proving and logic programming.
This general term indexing family is used
  to efficiently store and retrieve rewrite rules
  applicable to a term,
  and other widely-used techniques include
 substitution trees~\cite{stree}, context trees~\cite{ctree}, and path indexing~\cite{McCuneWilliam}.
The combination of discrimination trees,
  monoidal parsing, and monoidal hashing used in \name is,
  to the best of our understanding, novel,
  and may be useful in other term indexing contexts.
\section{Conclusion and Future Work}
Shortcut memoization is an approach to incremental computing 
  that composes and applies rewrite rules
  to speed up repetitive execution.
Our notion of rule composition,
  which gives the most general yet semantically valid shortcut,
  is mechanized in Agda including proofs of all theorems.
We then introduced \name,
  a purely functional programming language
  that compiles into an efficient abstract machine representation
  which can learn shortcuts by observing executions,
  store those shortcuts in a compact discrimination tree,
  and apply those shortcuts to future, similar program states.
This discrimination tree is accelerated
  using a custom representation of terms
  augmented by monoidal parsing and monoidal hashing.
Our evaluation shows that \name achieves
  a speedup of \hazelTotalSpeedup~
  on a set of program traces from a live programming setting.
Additionally,
  in a symbolic differentiation case study,
  \name is able to exploit shortcuts
  \emph{within a single execution},
  obtaining a speedup of \arithTotalSpeedup~
  thanks to repetitive program structures.

In future work, we hope to enrich the input to \name
  with \textit{theories}, by allowing equations to be specified
  in addition to reduction rules.
For example, the commutativity of addition
  could be incorporated into the rule matching procedure,
  resulting in more cases where a shortcut is applicable.
This may benefit from the use of e-graphs~\cite{egraph}.
\name's data structures could also be improved
  by experimenting with succinct data structures
  and term indexing data structures,
  resulting in larger speedups.
Finally, we plan to add a cache eviction policy
  and experiment with persistent caches,
  resulting in less memory and larger speedups
  across executions.

\section*{Data-Availability Statement}
We will submit our artifact for artifact evaluation. 
Our artifact contains an Agda mechanization of our main theorems, 
  the implementation of \name and our evaluation case studies. 
Upon executing the artifact, 
  the evaluator should see that Agda formally verifies our theorems
  and also generate a web page with plots from the evaluation 
  (albeit with slightly different numbers 
  due to the difference in execution environments).

\bibliographystyle{ACM-Reference-Format}
\bibliography{refs}

@inproceedings{tanimoto,
author = {Tanimoto, Steven L.},
title = {A perspective on the evolution of live programming},
year = {2013},
isbn = {9781467362658},
publisher = {IEEE Press},
booktitle = {Proceedings of the 1st International Workshop on Live Programming},
pages = {31–34},
numpages = {4},
location = {San Francisco, California},
series = {LIVE '13}
}

@inproceedings{szabo2016inca,
  title={Inca: A dsl for the definition of incremental program analyses},
  author={Szab{\'o}, Tam{\'a}s and Erdweg, Sebastian and Voelter, Markus},
  booktitle={Proceedings of the 31st IEEE/ACM International Conference on Automated Software Engineering},
  pages={320--331},
  year={2016}
}

@inproceedings{NaikAW06,
  author       = {Mayur Naik and
                  Alex Aiken and
                  John Whaley},
  title        = {Effective Static Race Detection for {Java}},
  booktitle    = {{ACM} {SIGPLAN} Conference on Programming
                  Language Design and Implementation},
  pages        = {308--319},
  publisher    = {{ACM}},
  year         = {2006},
  url          = {https://doi.org/10.1145/1133981.1134018},
  doi          = {10.1145/1133981.1134018},
  bibsource    = {dblp computer science bibliography, https://dblp.org}
}

@article{zhou2025incremental,
  title={An Incremental Algorithm for Algebraic Program Analysis},
  author={Zhou, Chenyu and Fang, Yuzhou and Wang, Jingbo and Wang, Chao},
  journal={Proceedings of the ACM on Programming Languages},
  volume={9},
  number={POPL},
  pages={1934--1961},
  year={2025},
  publisher={ACM New York, NY, USA}
}

@article{PollockS89,
  author       = {Lori L. Pollock and
                  Mary Lou Soffa},
  title        = {An Incremental Version of Iterative Data Flow Analysis},
  journal      = {{IEEE} Trans. Software Eng.},
  volume       = {15},
  number       = {12},
  pages        = {1537--1549},
  year         = {1989},
  url          = {https://doi.org/10.1109/32.58766},
  doi          = {10.1109/32.58766},
  bibsource    = {dblp computer science bibliography, https://dblp.org}
}

@article{nougrahiya2025incidfa,
  title={IncIDFA: An efficient and generic algorithm for incremental iterative dataflow analysis},
  author={Nougrahiya, Aman and Nandivada, V Krishna},
  journal={Proceedings of the ACM on Programming Languages},
  volume={9},
  number={OOPSLA1},
  pages={617--648},
  year={2025},
  publisher={ACM New York, NY, USA}
}

@inproceedings{ZhaoSRS21,
  author       = {David Zhao and
                  Pavle Subotic and
                  Mukund Raghothaman and
                  Bernhard Scholz},
  title        = {Towards Elastic Incrementalization for Datalog},
  booktitle    = {International Symposium on Principles and Practice
                  of Declarative Programming},
  pages        = {20:1--20:16},
  publisher    = {{ACM}},
  year         = {2021},
  url          = {https://doi.org/10.1145/3479394.3479415},
  doi          = {10.1145/3479394.3479415},
  bibsource    = {dblp computer science bibliography, https://dblp.org}
}

@inproceedings{szabo2021incremental,
  title={Incremental whole-program analysis in Datalog with lattices},
  author={Szab{\'o}, Tam{\'a}s and Erdweg, Sebastian and Bergmann, G{\'a}bor},
  booktitle={Proceedings of the 42nd ACM SIGPLAN International Conference on Programming Language Design and Implementation},
  pages={1--15},
  year={2021}
}

@inproceedings{ryder1983incremental,
  title={Incremental data flow analysis},
  author={Ryder, Barbara G},
  booktitle={Proceedings of the 10th ACM SIGACT-SIGPLAN symposium on Principles of programming languages},
  pages={167--176},
  year={1983}
}

@inproceedings{arzt2014reviser,
  title={Reviser: efficiently updating IDE-/IFDS-based data-flow analyses in response to incremental program changes},
  author={Arzt, Steven and Bodden, Eric},
  booktitle={Proceedings of the 36th International Conference on Software Engineering},
  pages={288--298},
  year={2014}
}

@article{popl19,
author = {Omar, Cyrus and Voysey, Ian and Chugh, Ravi and Hammer, Matthew A.},
title = {Live functional programming with typed holes},
year = {2019},
issue_date = {January 2019},
publisher = {Association for Computing Machinery},
address = {New York, NY, USA},
volume = {3},
number = {POPL},
url = {https://doi.org/10.1145/3290327},
doi = {10.1145/3290327},
journal = {Proc. ACM Program. Lang.},
month = jan,
articleno = {14},
numpages = {32}
}

@inproceedings{differential-dataflow,
  title={Differential Dataflow},
  author={Frank McSherry and Derek Gordon Murray and Rebecca Isaacs and Michael Isard},
  booktitle={Conference on Innovative Data Systems Research},
  year={2013},
  url={https://api.semanticscholar.org/CorpusID:18593675}
}

@article{incremental-layout,
author = {Kirisame, Marisa and Wang, Tiezhi and Panchekha, Pavel},
title = {Spineless Traversal for Layout Invalidation},
year = {2025},
issue_date = {June 2025},
publisher = {Association for Computing Machinery},
address = {New York, NY, USA},
volume = {9},
number = {PLDI},
url = {https://doi.org/10.1145/3729322},
doi = {10.1145/3729322},
journal = {Proc. ACM Program. Lang.},
month = jun,
articleno = {219},
numpages = {23}
}

@article{incremental-type-checking,
author = {Porter, Thomas J. and Kirisame, Marisa and Wei, Ivan and Panchekha, Pavel and Omar, Cyrus},
title = {Incremental Bidirectional Typing via Order Maintenance},
year = {2025},
issue_date = {October 2025},
publisher = {Association for Computing Machinery},
address = {New York, NY, USA},
volume = {9},
number = {OOPSLA2},
url = {https://doi.org/10.1145/3763117},
doi = {10.1145/3763117},
journal = {Proc. ACM Program. Lang.},
month = oct,
articleno = {339},
numpages = {28}
}

@article{build-systems-ala-carte,
author = {Mokhov, Andrey and Mitchell, Neil and Peyton Jones, Simon},
title = {Build systems \`{a} la carte},
year = {2018},
issue_date = {September 2018},
publisher = {Association for Computing Machinery},
address = {New York, NY, USA},
volume = {2},
number = {ICFP},
url = {https://doi.org/10.1145/3236774},
doi = {10.1145/3236774},
journal = {Proc. ACM Program. Lang.},
month = jul,
articleno = {79},
numpages = {29}
}

@article{incremental-parsing,
author = {Wagner, Tim A. and Graham, Susan L.},
title = {Efficient and flexible incremental parsing},
year = {1998},
issue_date = {Sept. 1998},
publisher = {Association for Computing Machinery},
address = {New York, NY, USA},
volume = {20},
number = {5},
issn = {0164-0925},
url = {https://doi.org/10.1145/293677.293678},
doi = {10.1145/293677.293678},
journal = {ACM Trans. Program. Lang. Syst.},
month = sep,
pages = {980–1013},
numpages = {34}
}

@inproceedings{incremental-compilation,
author = {Sathyanathan, Patrick W. and He, Wenlei and Tzen, Ten H.},
title = {Incremental whole program optimization and compilation},
year = {2017},
isbn = {9781509049318},
publisher = {IEEE Press},
booktitle = {Proceedings of the 2017 International Symposium on Code Generation and Optimization},
pages = {221–232},
numpages = {12},
location = {Austin, USA},
series = {CGO '17}
}

@inbook{spreadsheets,
author = {Abraham, Robin and Burnett, Margaret and Erwig, Martin},
year = {2009},
month = {03},
pages = {},
title = {Spreadsheet Programming},
isbn = {9780470050118},
doi = {10.1002/9780470050118.ecse415}
}

@article{incremental-linking,
author = {Quong, Russell W. and Linton, Mark A.},
title = {Linking programs incrementally},
year = {1991},
issue_date = {Jan. 1991},
publisher = {Association for Computing Machinery},
address = {New York, NY, USA},
volume = {13},
number = {1},
issn = {0164-0925},
url = {https://doi.org/10.1145/114005.102804},
doi = {10.1145/114005.102804},
journal = {ACM Trans. Program. Lang. Syst.},
month = jan,
pages = {1–20},
numpages = {20}
}

@misc{vercel,
  title        = {Incremental Static Regeneration},
  author       = {{Vercel}},
  year         = {2026},
  howpublished = {\url{https://vercel.com/docs/incremental-static-regeneration}},
  note         = {Vercel Documentation. Last updated January 27, 2026}
}

@article{ipython,
author = {Perez, Fernando and Granger, Brian E.},
title = {IPython: A System for Interactive Scientific Computing},
year = {2007},
issue_date = {May 2007},
publisher = {IEEE Educational Activities Department},
address = {USA},
volume = {9},
number = {3},
issn = {1521-9615},
url = {https://doi.org/10.1109/MCSE.2007.53},
doi = {10.1109/MCSE.2007.53},
journal = {Computing in Science and Engg.},
month = may,
pages = {21–29},
numpages = {9}
}

@article{safe-for-space,
author = {Shao, Zhong and Appel, Andrew W.},
title = {Efficient and safe-for-space closure conversion},
year = {2000},
issue_date = {Jan. 2000},
publisher = {Association for Computing Machinery},
address = {New York, NY, USA},
volume = {22},
number = {1},
issn = {0164-0925},
url = {https://doi.org/10.1145/345099.345125},
doi = {10.1145/345099.345125},
journal = {ACM Trans. Program. Lang. Syst.},
month = jan,
pages = {129–161},
numpages = {33}
}

@article{SAC,
author = {Acar, Umut A. and Blelloch, Guy E. and Harper, Robert},
title = {Adaptive functional programming},
year = {2006},
issue_date = {November 2006},
publisher = {Association for Computing Machinery},
address = {New York, NY, USA},
volume = {28},
number = {6},
issn = {0164-0925},
url = {https://doi.org/10.1145/1186632.1186634},
doi = {10.1145/1186632.1186634},
journal = {ACM Trans. Program. Lang. Syst.},
month = nov,
pages = {990–1034},
numpages = {45}
}

@inproceedings{adapton,
author = {Hammer, Matthew A. and Phang, Khoo Yit and Hicks, Michael and Foster, Jeffrey S.},
title = {Adapton: composable, demand-driven incremental computation},
year = {2014},
isbn = {9781450327848},
publisher = {Association for Computing Machinery},
address = {New York, NY, USA},
url = {https://doi.org/10.1145/2594291.2594324},
doi = {10.1145/2594291.2594324},
booktitle = {Proceedings of the 35th ACM SIGPLAN Conference on Programming Language Design and Implementation},
pages = {156–166},
numpages = {11},
location = {Edinburgh, United Kingdom},
series = {PLDI '14}
}

@article{nominal-adapton,
author = {Hammer, Matthew A. and Dunfield, Jana and Headley, Kyle and Labich, Nicholas and Foster, Jeffrey S. and Hicks, Michael and Van Horn, David},
title = {Incremental computation with names},
year = {2015},
issue_date = {October 2015},
publisher = {Association for Computing Machinery},
address = {New York, NY, USA},
volume = {50},
number = {10},
issn = {0362-1340},
url = {https://doi.org/10.1145/2858965.2814305},
doi = {10.1145/2858965.2814305},
journal = {SIGPLAN Not.},
month = oct,
pages = {748–766},
numpages = {19}
}

@inproceedings{Acar2003AdaptiveM,
  title={Adaptive Memoization},
  author={Umut A. Acar and Guy E. Blelloch and Robert Harper},
  year={2003},
  url={https://api.semanticscholar.org/CorpusID:6661244}
}

@inproceedings{vlhcc25,
  author={Rein, Patrick and Ramson, Stefan and Beckmann, Tom and Hirschfeld, Robert},
  booktitle={2025 IEEE Symposium on Visual Languages and Human-Centric Computing (VL/HCC)}, 
  title={An Information Foraging Interpretation of Liveness}, 
  year={2025},
  volume={},
  number={},
  pages={128-138},
  doi={10.1109/VL-HCC65237.2025.00022}
}

@inproceedings{Sundaresh1992IncrementalCV,
  title={Incremental computation via partial evaluation},
  author={R. S. Sundaresh},
  year={1992},
  url={https://api.semanticscholar.org/CorpusID:209401075}
}

@inproceedings{Gawrychowski2015OptimalDS,
  title={Optimal Dynamic Strings},
  author={Paweł Gawrychowski and Adam Karczmarz and Tomasz Kociumaka and Jakub Lacki and Piotr Sankowski},
  booktitle={ACM-SIAM Symposium on Discrete Algorithms},
  year={2015},
  url={https://api.semanticscholar.org/CorpusID:1496873}
}

@article{Sadakane2007SuccinctDS,
  title={Succinct data structures for flexible text retrieval systems},
  author={Kunihiko Sadakane},
  journal={J. Discrete Algorithms},
  year={2007},
  volume={5},
  pages={12-22},
  url={https://api.semanticscholar.org/CorpusID:15923596}
}

@article{He2010SuccinctRO,
  title={Succinct Representations of Dynamic Strings},
  author={Meng He and J. Ian Munro},
  journal={ArXiv},
  year={2010},
  volume={abs/1005.4652},
  url={https://api.semanticscholar.org/CorpusID:6157989}
}

@article{Quiring2024DerivingWD,
  title={Deriving with Derivatives: Optimizing Incremental Fixpoints for Higher-Order Flow Analysis},
  author={Benjamin Quiring and David Van Horn},
  journal={Proceedings of the ACM on Programming Languages},
  year={2024},
  volume={8},
  pages={728 - 755},
  url={https://api.semanticscholar.org/CorpusID:272071011}
}

@inproceedings{Zhu2004SymbolicPA,
  title={Symbolic pointer analysis revisited},
  author={Jianwen Zhu and Silvian Calman},
  booktitle={ACM-SIGPLAN Symposium on Programming Language Design and Implementation},
  year={2004},
  url={https://api.semanticscholar.org/CorpusID:5905331}
}

@article{McCuneWilliam,
author = {McCune, William},
title = {Experiments with discrimination-tree indexing and path indexing for term retrieval},
year = {1992},
issue_date = {Oct. 1992},
publisher = {Springer-Verlag},
address = {Berlin, Heidelberg},
volume = {9},
number = {2},
issn = {0168-7433},
url = {https://doi.org/10.1007/BF00245458},
doi = {10.1007/BF00245458},
journal = {J. Autom. Reason.},
month = oct,
pages = {147–167},
numpages = {21}
}

@inproceedings{stree,
  title={Substitution Tree Indexing},
  author={Peter Graf},
  booktitle={International Conference on Rewriting Techniques and Applications},
  year={1995},
  url={https://api.semanticscholar.org/CorpusID:29690319}
}

@article{ctree,
author = {Ganzinger, Harald and Nieuwenhuis, Robert and Nivela, Pilar},
title = {Fast Term Indexing with Coded Context Trees},
year = {2004},
issue_date = {February 2004},
publisher = {Springer-Verlag},
address = {Berlin, Heidelberg},
volume = {32},
number = {2},
issn = {0168-7433},
url = {https://doi.org/10.1023/B:JARS.0000029963.64213.ac},
doi = {10.1023/B:JARS.0000029963.64213.ac},
journal = {J. Autom. Reason.},
month = feb,
pages = {103–120},
numpages = {18}
}

@inproceedings{function-caching,
author = {Pugh, W. and Teitelbaum, T.},
title = {Incremental computation via function caching},
year = {1989},
isbn = {0897912942},
publisher = {Association for Computing Machinery},
address = {New York, NY, USA},
url = {https://doi.org/10.1145/75277.75305},
doi = {10.1145/75277.75305},
booktitle = {Proceedings of the 16th ACM SIGPLAN-SIGACT Symposium on Principles of Programming Languages},
pages = {315–328},
numpages = {14},
location = {Austin, Texas, USA},
series = {POPL '89}
}

@misc{incremental-lc,
      title={A Theory of Changes for Higher-Order Languages - Incrementalizing lambda-Calculi by Static Differentiation}, 
      author={Yufei Cai and Paolo G. Giarrusso and Tillmann Rendel and Klaus Ostermann},
      year={2013},
      eprint={1312.0658},
      archivePrefix={arXiv},
      primaryClass={cs.PL},
      url={https://arxiv.org/abs/1312.0658}, 
}

@article{patch-reconciliation,
author = {Ziegler, Parker and Lubin, Justin and Chasins, Sarah E.},
title = {Fast Direct Manipulation Programming with Patch-Reconciliation Correspondence},
year = {2025},
issue_date = {June 2025},
publisher = {Association for Computing Machinery},
address = {New York, NY, USA},
volume = {9},
number = {PLDI},
url = {https://doi.org/10.1145/3729278},
doi = {10.1145/3729278},
journal = {Proc. ACM Program. Lang.},
month = jun,
articleno = {175},
numpages = {26}
}

@inproceedings{precise-dependencies,
author = {Heydon, Allan and Levin, Roy and Yu, Yuan},
title = {Caching function calls using precise dependencies},
year = {2000},
isbn = {1581131992},
publisher = {Association for Computing Machinery},
address = {New York, NY, USA},
url = {https://doi.org/10.1145/349299.349341},
doi = {10.1145/349299.349341},
booktitle = {Proceedings of the ACM SIGPLAN 2000 Conference on Programming Language Design and Implementation},
pages = {311–320},
numpages = {10},
location = {Vancouver, British Columbia, Canada},
series = {PLDI '00}
}

@article{caching-dependencies,
author = {Abadi, Mart\'{\i}n and Lampson, Butler and L\'{e}vy, Jean-Jacques},
title = {Analysis and caching of dependencies},
year = {1996},
issue_date = {June 15, 1996},
publisher = {Association for Computing Machinery},
address = {New York, NY, USA},
volume = {31},
number = {6},
issn = {0362-1340},
url = {https://doi.org/10.1145/232629.232638},
doi = {10.1145/232629.232638},
journal = {SIGPLAN Not.},
month = jun,
pages = {83–91},
numpages = {9}
}

@inproceedings{cek,
  title={Control operators, the SECD-machine, and the $\lambda$-calculus},
  author={Matthias Felleisen and Daniel P. Friedman},
  booktitle={Formal Description of Programming Concepts},
  year={1987},
  url={https://api.semanticscholar.org/CorpusID:57760323}
}

@article{egraph,
author = {Nelson, Greg and Oppen, Derek C.},
title = {Fast Decision Procedures Based on Congruence Closure},
year = {1980},
issue_date = {April 1980},
publisher = {Association for Computing Machinery},
address = {New York, NY, USA},
volume = {27},
number = {2},
issn = {0004-5411},
url = {https://doi.org/10.1145/322186.322198},
doi = {10.1145/322186.322198},
journal = {J. ACM},
month = apr,
pages = {356–364},
numpages = {9}
}

@article{hazel-OM,
author = {Porter, Thomas J. and Kirisame, Marisa and Wei, Ivan and Panchekha, Pavel and Omar, Cyrus},
title = {Incremental Bidirectional Typing via Order Maintenance},
year = {2025},
issue_date = {October 2025},
publisher = {Association for Computing Machinery},
address = {New York, NY, USA},
volume = {9},
number = {OOPSLA2},
url = {https://doi.org/10.1145/3763117},
doi = {10.1145/3763117},
journal = {Proc. ACM Program. Lang.},
month = oct,
articleno = {339},
numpages = {28}
}

@article{hazel-EC,
author = {Moon, David and Blinn, Andrew and Porter, Thomas J. and Omar, Cyrus},
title = {Syntactic Completions with Material Obligations},
year = {2025},
issue_date = {October 2025},
publisher = {Association for Computing Machinery},
address = {New York, NY, USA},
volume = {9},
number = {OOPSLA2},
url = {https://doi.org/10.1145/3763182},
doi = {10.1145/3763182},
journal = {Proc. ACM Program. Lang.},
month = oct,
articleno = {404},
numpages = {27}
}
\end{document}